\documentclass[11pt]{article}

\usepackage[papersize={8.5in,11in},top=1in,left=1in,right=1in,bottom=1in]{geometry}
\usepackage{ifthen}
\usepackage{xspace}
\usepackage{amsmath}
\usepackage{bm}
\usepackage[section]{definitions}
\usepackage{tikz}
\usetikzlibrary{arrows}

\newtheorem{conjecture}{Conjecture}

\newcommand{\ALLCANDS}{\ensuremath{X}\xspace}
\newcommand{\ALLCANDSP}{\ensuremath{X'}\xspace}
\newcommand{\ALLVOTERS}{\ensuremath{V}\xspace}
\newcommand{\ALLPREFS}{\ensuremath{\mathcal{P}}\xspace}

\newcommand{\PREFDIST}{\ensuremath{\mathcal{D}}\xspace}

\newcommand{\PREF}[1][]{\ensuremath{%
    \ifthenelse{\equal{#1}{}}{\succ}{\succ_{#1}}}\xspace}
\newcommand{\Pref}[3][]{\ensuremath{#2 \PREF[#1] #3}\xspace}
\newcommand{\PREFP}[1][]{\ensuremath{%
    \ifthenelse{\equal{#1}{}}{\succ'}{\succ'_{#1}}}\xspace}
\newcommand{\PrefP}[3][]{\ensuremath{#2 \PREFP[#1] #3}\xspace}
\newcommand{\WPREF}[1][]{\ensuremath{%
    \ifthenelse{\equal{#1}{}}{\succeq}{\succeq_{#1}}}\xspace}
\newcommand{\WPref}[3][]{\ensuremath{#2 \WPREF[#1] #3}\xspace}
\newcommand{\CandFirst}[2]{\ensuremath{[#1 \succ #2]}\xspace}
\newcommand{\CandLast}[2]{\ensuremath{[#2 \succ #1]}\xspace}

\newcommand{\SCRULE}{\ensuremath{f}\xspace}
\newcommand{\SCRule}[1]{\ensuremath{\SCRULE(#1)}\xspace}

\newcommand{\Cost}[1]{\ensuremath{C(#1)}\xspace}
\newcommand{\Distortion}[1]{\ensuremath{\rho(#1)}\xspace}

\newcommand{\DIST}{\ensuremath{d}\xspace}
\newcommand{\Dist}[2]{\ensuremath{\DIST(#1,#2)}\xspace}
\newcommand{\Legal}[2]{\ensuremath{#1 \sim #2}\xspace}

\newcommand{\OPT}[1][]{\ensuremath{%
    \ifthenelse{\equal{#1}{}}{\LPC}{\LPC_{#1}}}\xspace}
\newcommand{\WINNER}{\ensuremath{w}\xspace}

\newcommand{\Lpm}[2]{\ensuremath{d_{#2,#1}}\xspace}
\newcommand{\LPC}{\ensuremath{x^*}\xspace}
\newcommand{\DNo}[1]{\ensuremath{\alpha_{#1}}\xspace}
\newcommand{\DCo}[3]{\ensuremath{\phi^{(#1)}_{#2,#3}}\xspace}
\newcommand{\DTe}[4]{\ensuremath{\psi^{(#1,#2)}_{#3,#4}}\xspace}
\newcommand{\OptDist}[2]{\ensuremath{c_{#1,#2}}\xspace}
\newcommand{\WOptDist}[1]{\ensuremath{\bar{c}_{#1}}\xspace}
\newcommand{\MATCH}{\ensuremath{\mu}\xspace}
\newcommand{\Match}[1]{\ensuremath{\MATCH_{#1}}\xspace}
\newcommand{\Inter}[1]{\ensuremath{\tilde{x}_{#1}}\xspace}

\newcommand{\PREFFRAC}{\ensuremath{p}\xspace}
\newcommand{\PrefFrac}[1]{\ensuremath{\PREFFRAC_{#1}}\xspace}

\newcommand{\FlCost}[2]{\ensuremath{\gamma^{(#1)}_{#2}}\xspace}

\newcommand{\Block}[2]{\ensuremath{\mathcal{B}_{#1,#2}}\xspace}

\newcommand{\Bip}[2]{\ensuremath{H_{#1,#2}}\xspace}
\newcommand{\BipP}[2]{\ensuremath{H'_{#1,#2}}\xspace}
\newcommand{\Neigh}[1]{\ensuremath{\Gamma(#1)}\xspace}

\newcommand{\COMPG}{\ensuremath{\mathcal{G}}\xspace}
\newcommand{\CompG}[2]{\ensuremath{\mathcal{G}(#1,#2)}\xspace}
\newcommand{\COMPGP}{\ensuremath{\mathcal{G}'}\xspace}

\newcommand{\EWit}[2]{\ensuremath{Z_{#1,#2}}\xspace}
\newcommand{\EWitS}[1]{\ensuremath{Z_{#1}}\xspace}

\newcommand{\CandNode}[1]{\ensuremath{y_{#1}}\xspace}
\newcommand{\SetNode}[1]{\ensuremath{a_{#1}}\xspace}
\newcommand{\ComplNode}[1]{\ensuremath{b_{#1}}\xspace}
\newcommand{\CandNodes}[1]{\ensuremath{Y_{#1}}\xspace}
\newcommand{\SetNodes}[1]{\ensuremath{A_{#1}}\xspace}
\newcommand{\ComplNodes}[1]{\ensuremath{B_{#1}}\xspace}

\newcommand{\IndPos}[2]{\ensuremath{\alpha(#1,#2)}\xspace}
\newcommand{\IndNeg}[2]{\ensuremath{\beta(#1,#2)}\xspace}

\newcommand{\ABM}{\textsc{All Bipartite Matchings}\xspace}
\newcommand{\CandCompGr}{Candidate Comparison Graph\xspace}

\begin{document}

\title{An Analysis Framework for Metric Voting based on LP Duality}
\author{David Kempe\\
University of Southern California}


\maketitle

\begin{abstract}
  Distortion-based analysis has established itself as a fruitful framework for comparing voting mechanisms. The assumption is that the $m$ voters and $n$ candidates are jointly embedded in an (unknown) metric space, and the voters submit rankings of candidates by non-decreasing distance from themselves. Based on the submitted rankings, the social choice rule chooses a winning candidate; the quality of the winner is the sum of the (unknown) distances to the voters.
Since it is missing the information about the actual distances, the rule's choice will in general be suboptimal, and the worst-case ratio between the cost of its chosen candidate and the optimal candidate is called the rule's \emph{distortion}. It was shown in prior work that every deterministic rule has distortion at least 3, while the Copeland rule and related rules guarantee distortion at most 5, and a very recent result gave a generalization of Copeland with distortion $2+\sqrt{5} \approx 4.236$.

We provide a framework based on LP-duality and flow interpretations of the dual which provides a simpler and more unified way for proving upper bounds on the distortion of social choice rules. Rather than having to reason about all possible metric spaces, to establish an upper bound, it is sufficient to exhibit a certain type of flow with small cost.
We illustrate the utility of this approach with three examples.
First, we give a fairly simple proof of a strong generalization of the upper bound of 5 on the distortion of Copeland, to social choice rules with short paths from the winning candidate to the optimal candidate in generalized weak preference graphs. A special case of this result recovers the recent $2+\sqrt{5}$ guarantee.
Next, we use this generalization to show that the Ranked Pairs and Schulze rules have distortion $\Theta(\sqrt{n})$.
Finally, our framework naturally suggests a combinatorial rule that is a strong candidate for achieving distortion 3, which had also been proposed in recent work. We prove that the distortion bound of 3 would follow from any of three combinatorial conjectures we formulate (and have verified by computer for $n \leq 7$ candidates).

\end{abstract}

\section{Introduction} \label{sec:introduction}
Voting is an important and widespread way for a group to choose one
out of multiple available candidate options.\footnote{%
  In this submission, we do not consider the equally important and
  widely studied problem of a group ranking all of the available options.}
The group could be a country, academic department, or other
organization,
and the $n$ candidate options they choose from could be courses of
action or human candidates. 
Typically, each voter submits a total order of all options,
called a \emph{ranking} or \emph{preference order}.
Based on all the submitted rankings, a \emph{social choice rule}
(or \emph{mechanism}) determines the winning option.

Different mechanisms will have different desirable and undesirable
properties, and it is important to articulate and analyze these
properties to guide an organization's choice of mechanisms.
The axiomatic approach, dating back at least several centuries
\cite{borda:elections,condorcet:essay}, articulates natural axioms
about the properties that the mapping from rankings to a winner should
satisfy, and has led to extensive work
(see, e.g., \cite{BCULP:social-choice} for an overview). 
Unfortunately, many of the key results are impossibility results,
in particular the famous Gibbard-Satterthwaite Theorem
\cite{gibbard:manipulation,satterthwaite:voting} showing that there is
no truthful mechanism satisfying very minimal additional properties.

An alternative that has gained much recent popularity,
in particular in the computer science community,
is to view social choice through the lens of
\emph{optimization and approximation}.
In this line of work (e.g., 
\cite{BCHLPS:utilitarian:distortion,caragiannis:procaccia:voting,procaccia:approximation:gibbard,procaccia:rosenschein:distortion}),
it is assumed that one can quantify the utility (or cost) that a voter
derives from a candidate.
These individual utilities or costs can then be aggregated into a
social welfare or cost, e.g., by taking the average or median.
The social welfare/cost captures how good of a choice a candidate is
for the voter population overall.

The problem with this approach, articulated clearly in
\cite{boutilier:rosenschein:incomplete,anshelevich:bhardwaj:postl},
is that voting mechanisms typically allow voters only to communicate a
\emph{ranking} of candidates, but not the actual utilities/costs;
furthermore, even if a mechanism provided a way to communicate
numerical scores, it is not clear that voters could compute or
estimate them accurately.
In other words, ``one can quantify'' is more of an abstract statement
than one referring to any decision maker involved in the process.
Thus, even though the voting mechanism must optimize a \emph{cardinal}
objective function, it only receives \emph{ordinal} information as
input, namely, for each voter, whether her\footnote{
  For ease of presentation, we use female pronouns for voters and
  male pronouns for candidates throughout.}
utility/cost for candidate $x$ is larger or smaller than that for
candidate $y$.

As a result, mechanisms must optimize the social welfare
\emph{robustly}, choosing a candidate that has high welfare regardless
of what the actual cardinal objective values are --- so long as they are
consistent with the reported ordinal rankings.
The \emph{distortion} of a mechanism is the worst-case ratio between
the welfare/cost of the mechanism's selected (based only on ordinal
information) candidate and the optimum (with full knowledge of the
cardinal values) candidate, over all possible inputs.
(Formal definitions of this concept and all other terms can be found
in Section~\ref{sec:preliminaries}.)

Our discussion so far has been in terms of general utilities/cost.
While some positive results can be obtained for fairly general classes
of utility functions (e.g.,
\cite{BCHLPS:utilitarian:distortion,caragiannis:procaccia:voting,procaccia:approximation:gibbard,procaccia:rosenschein:distortion}), 
stronger results are achievable when the functions take more specific
forms.
A particularly natural way of defining costs is in terms of a joint
\emph{metric space} defined on candidates and voters,
where the distance \Dist{v}{x} between voter $v$ and candidate $x$
captures their difference in opinion, and hence the cost.
Voters then rank candidates by non-decreasing distance from themselves.%
\footnote{Such distance-based rankings had been considered
in a large body of earlier work, e.g.,
\cite{black:rationale,black:committees-elections,downs:democracy,moulin:single-peak,merrill:grofman,barbera:gul:stacchetti,barbera:social-choice},
though most of the listed papers studied such rankings specifically
when the metric is the line; 
such preference orders are often called \emph{single-peaked}.}
The approach of using the distances explicitly as the cost objective
for optimization was proposed in \cite{anshelevich:bhardwaj:postl};
\cite{anshelevich:bhardwaj:elkind:postl:skowron} is an
expanded/improved journal version,
and \cite{anshelevich:ordinal} provides a broader overview of the area
and its results.
While \cite{anshelevich:bhardwaj:elkind:postl:skowron} consider both
the average and median of all voters' costs as the overall objective,
here, we focus solely on the average/total cost.

The main result of 
\cite{anshelevich:bhardwaj:postl,anshelevich:bhardwaj:elkind:postl:skowron}
is that under the model of metric costs,
many widely used voting rules
(including Plurality, Veto, Borda count, and others)
have distortion linear in the number of candidates or worse.
Furthermore, even with just 2 candidates and a 1-dimensional metric
space, \emph{every} deterministic voting mechanism has 
distortion at least 3.
On the positive side,
\cite{anshelevich:bhardwaj:postl,anshelevich:bhardwaj:elkind:postl:skowron}
show that any rule which always outputs a candidate from the uncovered
set of candidates has distortion at most 5,
for all metric spaces and numbers of candidates.
Uncovered sets are defined in terms of the tournament
graph $G$ on $n$ candidates in which the directed edge $(x,y)$ is
present iff a (weak) majority of voters prefer $x$ to $y$.
The uncovered set is the set of candidates that have
a directed path of length at most 2 in $G$ to every other candidate
(see \cite{moulin:choosing-tournament}).
Very recently, Munagala and Wang \cite{munagala:wang:improved}
gave a voting rule based on uncovered sets in a weighted tournament
graph which improves the upper bound from 5 to
$2+\sqrt{5} \approx 4.236$.

There is an obvious gap between the lower bound of 3 for the
distortion of every mechanism, and the upper bound of $2+\sqrt{5}$.
In the original version of \cite{anshelevich:bhardwaj:postl},
it was conjectured that a mechanism called \emph{Ranked Pairs}
(defined in Section~\ref{sec:preliminaries})
achieves a distortion of 3.
This conjecture was disproved by \cite{goel:krishnaswamy:munagala},
who showed a lower bound of 5 on the distortion of Ranked Pairs (and
the Schulze rule, also defined in Section~\ref{sec:preliminaries}).

The proof of the upper bound of 5, the recent upper bound of $2+\sqrt{5}$,
and many other proofs in the literature are based on reasoning about
all metric spaces that are consistent with assumed rankings.
They often involve intricate case distinctions and rather ad
hoc arguments.
So far, a more solid foundation and framework for distortion proofs
has been missing from the literature.

\subsection{Our Contribution}
Our main contribution, presented in Section~\ref{sec:primal-dual},
is an analysis framework based on LP duality and flows
for proving upper bounds on the metric distortion of voting mechanisms.
Our point of departure is a well-known linear program for the
following problem:
given the rankings of all voters, a winning candidate
(presumably selected by a mechanism) and an ``optimum'' candidate,
find a metric space maximizing the distortion of this choice;
that is, find a metric that makes the selected winner as expensive as
possible, subject to the ``optimum'' candidate having cost 1.%
\footnote{This approach can of course immediately be leveraged into an
  optimal polynomial-time voting mechanism;
  we discuss this more in Section~\ref{sec:optimal}.}
We show that the dual of the cost minimization LP can be interpreted
as a flow problem with an unusual objective function.
Using this framework, in order to show an upper bound on the metric
distortion of a particular mechanism,
rather than having to explicitly consider all possible metric spaces,
it is enough to exhibit a flow of small cost meeting certain demands.
We illustrate the power of this analysis framework with three applications.

First, in Section~\ref{sec:uncovered},
we give a strong generalization of the key lemmas from
\cite{anshelevich:bhardwaj:postl} (Theorem~7)
and \cite{munagala:wang:improved} (Lemma~3.7),
used to prove distortions of 5 and $2+\sqrt{5}$ for the respective
mechanisms under consideration.
The common idea of both is that when a large enough fraction of voters
prefer $x$ to $y$, and a large enough fraction prefer $y$ to $z$,
then the cost of $x$ can be bounded in terms of the cost of $z$.
Theorem~7 of \cite{anshelevich:bhardwaj:postl} is the special case
where both fractions are \half,
while Lemma~3.7 of \cite{munagala:wang:improved} is the case
when the first fraction is $\frac{3-\sqrt{5}}{2}$,
and the second is $\frac{\sqrt{5}-1}{2}$.
These bounds immediately imply the upper bounds on the distortion 
for any candidate in the uncovered set of a suitably defined
tournament graph.
We give a generalization to arbitrary chains of preferences,
and upper-bound the cost of $x_1$ in terms of the cost of $x_{\ell}$
when a \PrefFrac{i} fraction of voters prefer $x_i$ over $x_{i+1}$,
for each $i=1, \ldots, \ell-1$.
For the specific case when all $\PrefFrac{i} = \PREFFRAC$,
the bound can be stated very cleanly:
the cost of $x_1$ is at most $\frac{\ell}{\PREFFRAC} - 1$ times that
of $x_{\ell}$ if $\ell$ is even, 
and at most $\frac{\ell-1}{\PREFFRAC} + 1$ times that of $x_{\ell}$ if
$\ell$ is odd.
Our results fully recover and generalize the bounds of
\cite{anshelevich:bhardwaj:postl} and \cite{munagala:wang:improved}.
The generalization to longer path lengths can be useful in
analyzing voting mechanisms that are missing information.
This can happen if the environment restricts the
communication between voters and the mechanism,
so that parts of the rankings remain unknown,
as in \cite{DistortionCommunication}.
In fact, the results of Section~\ref{sec:uncovered} can be used to
significantly improve the upper bounds on the performance of
``Copeland-like'' mechanisms with missing information,
compared to the bounds in \cite{DistortionCommunication}.

As a direct application of this generalized bound,
in Section~\ref{sec:rp-schulze},
we resolve the distortion of the Ranked Pairs and Schulze rules
(defined in Section~\ref{sec:preliminaries}):
we show that both have distortion $\Theta(\sqrt{n})$.
The upper bound is a clean application of the lemma bounding
distortion via longer chains of preferences,
while the lower bound is obtained with a generalization of the
example which \cite{goel:krishnaswamy:munagala} used to
lower-bound the distortion of both rules by 5.
The distortion of both rules is thus significantly higher than
the distortions of 5 and $2+\sqrt{5}$ achieved by the uncovered set
mechanisms.
Understanding the distortion of the Schulze rule in
particular is of importance because it is widely used in practice.

As a third application, the flow interpretation naturally
suggests a candidate mechanism that might achieve 
distortion 3, which we present in Section~\ref{sec:combinatorial}.
The analysis points to a sufficient condition for distortion 3:
that for every given preference profile of the voters,
there be a candidate $x$ such that for all other candidates $y$,
a certain bipartite graph on the voters have a perfect matching.
In fact, the mechanism itself can be phrased in this terminology,
leading to a purely combinatorial polynomial-time mechanism.

This mechanism was independently discovered and presented in
\cite{munagala:wang:improved}.
In \cite{munagala:wang:improved}, it is also shown --- again with a
case distinction proof over metric spaces --- that if such a
candidate $x$ exists, the mechanism guarantees distortion 3.
Our duality framework gives a cleaner and simpler proof
of this fact.
The main question is then whether the desired candidate $x$ always
exists.

Munagala and Wang \cite{munagala:wang:improved} conjecture --- as do
we --- that it does.
They phrase a conjecture which is essentially a restatement of the
fact that the algorithm succeeds in finding a candidate $x$.
In Section~\ref{sec:combinatorial}, we present a slight rephrasing of
this conjecture, along with two more very different-looking (in fact,
much more self-contained) conjectures, each of which would resolve the
question positively, i.e., establish a distortion of 3.
One of the two new conjectures is phrased in terms of certain preferences
between candidates and sets under randomly drawn preference orders,
while another talks about cycles in certain induced subsets of a type
of graph we define.
The fact that they are sufficient to establish distortion 3 is based
on Hall's Marriage Theorem for bipartite graphs.
We have verified the conjecture by hand for $n \leq 4$ candidates,
and using exhaustive computer search for $n \leq 7$.
Resolving any of the three conjectures positively would answer the key
open question of the field of metric voting,
closing the gap between the upper bound of $2+\sqrt{5}$
on the best distortion of any deterministic mechanism,
and the lower bound of 3. 

\subsection{Additional Related Work}

The observation that mechanisms may have to optimize a cardinal
objective function while only given ordinal information (i.e.,
rankings) extends beyond just voting mechanisms, to more general
problems.
See, e.g., \cite{anshelevich:sekar:blind,anshelevich:ordinal} for
results on other optimization problems under ordinal information.

The lower bound of 3 on the distortion of any mechanism is based on
worst-case input instances.
Better bounds can be obtained when additional assumptions are placed
on the instances.
As one example, 
\cite{anshelevich:postl:randomized,gross:anshelevich:xia:agree}
show that when instances are \emph{decisive}, in the sense that each
voter has a candidate she strongly prefers over all others,
better upper bounds on the distortion are obtained.
As another example, when the candidates are drawn i.i.d.~from the set
of all voters, \cite{OfThePeople} gives improved constant distortion
bounds in the case of two candidates, while
\cite{BordaRepresentative} shows that many position-based scoring
rules now achieve constant distortion (instead of linear).

The lower bound of 3 on the distortion of voting mechanisms only
applies to deterministic mechanisms.
Randomization can lead to lower distortion \cite{anshelevich:postl:randomized}.
For example, it is known that the Randomized Dictatorship mechanism,
which outputs the first choice of a uniformly random voter,
has distortion strictly smaller than 3.

Our work ignores the issue of incentives, i.e., whether voters
truthfully report their preferences.
The connection between strategy proofness and distortion in metric
voting is studied in \cite{feldman:fiat:golomb}.

The use of LP duality for analyzing the performance of optimization
algorithms has a long history, e.g., in approximation algorithms
(see \cite{vazirani:approximation-algorithms}).
Another more recent example is the duality framework of Cai, Devanur,
and Weinberg \cite{cai:devanur:weinberg:duality}
(see also references in \cite{cai:devanur:weinberg:duality} to prior,
less general, work) for analyzing the revenue of Bayesian Incentive
Compatible mechanisms.
In their case as well, dual solutions can be interpreted as flows,
and Cai et al.~obtain performance guarantees by exhibiting particular
types of ``canonical'' flows that can be interpreted as witnesses for
the revenue guarantees.
While this work and ours have the use of duality, and the
interpretation as flows, in common, the specific technical details are
very different.

\section{Preliminaries} \label{sec:preliminaries}
\subsection{Voters, Candidates, and Social Choice Rules}

An \emph{instance} $(\ALLCANDS, \ALLPREFS)$ consists of a set
of $n$ candidates \ALLCANDS, and the voters' preferences \ALLPREFS
among these candidates.
Candidates will always be denoted by lowercase letters $\WINNER, x, y,
z$ (and their variations), with \WINNER 
specifically reserved for a candidate chosen as winner by a mechanism
(which will be clear from the context).
Sets of candidates are denoted by uppercase letters $X, Y, Z$.
The $m$ voters are denoted by $v,v'$ and variants thereof,
and the set of all voters is \ALLVOTERS.

Each voter $v$ has a \emph{total order} (or \emph{preference order} or
\emph{ranking} --- we use the three terms interchangeably) \PREF[v]
over the $n$ candidates.
We write \Pref[v]{x}{y} to denote that $v$ (strictly) prefers $x$ over $y$,
and \WPref[v]{x}{y} to denote that $v$ weakly prefers $x$ over $y$;
the difference is that the latter allows $x=y$.
We extend this notation to sets, writing, for instance,
\Pref[v]{Y}{Z} to denote that $v$ (strictly) prefers all candidates in
$Y$ over all candidates in $Z$.
We write $\CandFirst{x}{Y} = \Set{v \in \ALLVOTERS}{\Pref[v]{x}{Y}}$
for the set of voters who rank $x$ strictly ahead of all candidates in
$Y$,
and $\CandLast{x}{Y} = \Set{v \in \ALLVOTERS}{\Pref[v]{Y}{x}}$
for the set of voters who rank $x$ strictly behind all candidates in $Y$.

A \emph{vote profile} \ALLPREFS is the vector of the rankings of all
voters $\ALLPREFS = (\PREF[v])_{v \in \ALLVOTERS}$.
A \emph{social choice rule}
(we use the term \emph{mechanism} interchangeably) 
$\SCRULE : (\ALLCANDS, \ALLPREFS) \mapsto \WINNER$
is given the rankings of all voters, i.e., \ALLPREFS,
and deterministically produces as output one \emph{winning} candidate
$\WINNER = \SCRule{\ALLCANDS, \ALLPREFS} \in \ALLCANDS$.

\subsection{(Pseudo-)Metric Space and Distortion}

The voter preferences are assumed to be derived from distances between
voters and candidates.
The distance \Dist{v}{x} between voter $v$ and candidate $x$ captures
how similar their positions on key issues are.
The distances \DIST form a \emph{pseudo-metric}, i.e.,
they are non-negative and satisfy the triangle inequality%
\footnote{Distances between pairs of voters, or between pairs of
  candidates, could be defined using shortest-path distances;
  however, they are irrelevant for our analysis.
  \emph{Symmetry}, another defining property of a pseudo-metric,
  would arise automatically when using this definition.}
$\Dist{v}{x} \leq \Dist{v}{y} + \Dist{v'}{y} + \Dist{v'}{x}$
for all voters $v, v'$ and candidates $x, y$.

A vote profile \ALLPREFS is \emph{consistent} with the
pseudo-metric \DIST if and only if each voter ranks the candidates by 
non-decreasing distance from herself;
that is, if \Pref[v]{x}{y} whenever $\Dist{v}{x} < \Dist{v}{y}$.
When \ALLPREFS is consistent with \DIST, we write \Legal{\DIST}{\ALLPREFS}.
If there are ties among distances,
several vote profiles will be consistent with \DIST.

\begin{definition}[Social Cost, Distortion]
\begin{enumerate}
\item The \emph{social cost} of candidate $x$ is the sum of distances
  from $x$ to all voters: $\Cost{x} = \sum_{v} \Dist{v}{x}$.
\item A candidate is an \emph{optimum candidate} iff he minimizes%
  \footnote{There could be multiple optimum candidates --- for our
    analysis, it will never matter which of them is designated as
    ``the'' optimum.}
  the social cost: $\OPT[\DIST] \in \argmin_{x \in \ALLCANDS} \Cost{x}$.
\item The \emph{distortion} of a mechanism \SCRULE is the largest
  possible ratio between the cost of the candidate chosen by \SCRULE,
  and the optimal
  (with respect to the pseudo-metric \DIST, which \SCRULE does not know)
  candidate \OPT[\DIST]:
\[
  \Distortion{\SCRULE} \; = \;
  \max_{\ALLPREFS} \sup_{\DIST: \Legal{\DIST}{\ALLPREFS}}
  \frac{\Cost{\SCRule{\ALLCANDS, \ALLPREFS}}}{\Cost{\OPT[\DIST]}}.
\]
\end{enumerate}
\end{definition}

The main cause for (large) distortion is that while
the social choice rule knows the voter preferences \ALLPREFS,
it does not know the pseudo-metric \DIST.
We can think of the pseudo-metric \DIST as being chosen
adversarially,
based on the winning candidate \WINNER = \SCRule{\ALLPREFS} chosen by
the mechanism.
However, the adversary is constrained by having to ensure that \DIST
is consistent with \ALLPREFS.

\subsection{Ranked Pairs and the Schulze Rule}
In Section~\ref{sec:rp-schulze}, we will characterize the distortion
of the Ranked Pairs and Schulze Rules.
We briefly review these rules here.
Both are based on a weighted directed graph on the set of candidates
\ALLCANDS.
The weight $p_{x,y}$ of the edge from candidate $x$ to $y$
is the fraction of voters who have \Pref{x}{y}.
As a result, $p_{x,y} + p_{y,x} = 1$ for all $x, y$.

In Ranked Pairs \cite{tideman:independence-of-clones},
the (ordered) pairs $(x,y)$ are considered in
non-increasing order of $p_{x,y}$.
When the pair $(x,y)$ is considered,
the directed edge $(x,y)$ is inserted into the graph if and only if
doing so creates no cycle.
When the insertion process terminates, the graph has a unique source
node, which is returned as the winner.

In the Schulze Method \cite{schulze:single-winner-election-method},
a directed weighted graph is created in which
each ordered pair $(x,y)$ has an edge with weight $p_{x,y}$.
Then, for each pair $(x,y)$, let $s_{x,y}$ be the width of the widest
path from $x$ to $y$, that is, the largest $p$ such that there is a
path from $x$ to $y$ on which all edges $(x',y')$ have
$p_{x',y'} \geq p$.
It has been shown \cite{schulze:single-winner-election-method}
that there is a candidate node $x$ such
that $s_{x,y} \geq s_{y,x}$ for all other candidates $y$.
Any such candidate $x$ is returned as the winner.

For the purposes of our analysis, the only property of these methods
that matters is captured by the following lemma, which is well known.
(We prove it only for completeness.)

\begin{lemma} \label{lem:rp-schulze-basic}
  Let \WINNER be the candidate selected by the rule
  (either Ranked Pairs or Schulze),
  and $y$ any other candidate.
  Then, there exists a $p$ and a sequence of (distinct) candidates
  $x_1 = \WINNER, x_2, \ldots, x_{\ell} = y$ with the property that
  at least a $p$ fraction of voters prefer $x_i$ over $x_{i+1}$
  (for each $i$),
  and at most a $p$ fraction of voters prefer $y$ over \WINNER.
\end{lemma}

\begin{proof}
  For the Ranked Pairs rule, because \WINNER was selected, it
  has no incoming edges in the DAG that is constructed.
  In particular, this means that Ranked Pairs did not insert the
  edge $(y, \WINNER)$, so when it was considered for insertion,
  it would have caused a cycle, meaning that there was a path from
  \WINNER to $y$ all of whose edges had been inserted earlier.
  By the definition of the Ranked Pairs insertion order, this
  means that all edges on this path had a higher fraction of voters
  agreeing with them, giving us the path claimed above.

  For the Schulze rule, recall that the winner has the
  property that $s_{\WINNER,x} \geq s_{x,\WINNER}$ for all
  candidates $x$. Let $p=s_{\WINNER,y}$.
  Then, there is a path from \WINNER to $y$ in which each edge
  corresponds to a preference by at least a $p$ fraction of voters.
  On the other hand, because $(y, \WINNER)$ is a path from $y$
  to \WINNER, at most an $s_{y,\WINNER} \leq s_{\WINNER,y}$
  fraction of voters can prefer $y$ to \WINNER.
\end{proof}

\section{The LP Duality Approach and Flows} \label{sec:primal-dual}
In this section, we develop the key tool for our analysis:
the dual linear program for distortion in metric voting.

The voters' preferences $\ALLPREFS = (\PREF[v])_v$ are given.
Let \WINNER be a candidate that the mechanism is considering as a
potential winner, and \LPC the optimal candidate.
Following \cite{anshelevich:bhardwaj:elkind:postl:skowron,goel:krishnaswamy:munagala},
we phrase the adversary's problem of finding the distortion-maximizing
metric as a linear program whose variables \Lpm{x}{v} denote
distances between voters $v$ and candidates $x$.
These distances must be non-negative, obey the triangle inequality,
and be consistent with the reported preferences of the voters.
The objective is to maximize the distortion,
i.e., the ratio between the cost of \WINNER and the cost of \LPC.

\begin{LP}[eqn:primal-lp]{Maximize}{\sum_v \Lpm{\WINNER}{v}}
  \Lpm{x}{v} \leq \Lpm{x}{v'} + \Lpm{y}{v'} + \Lpm{y}{v}
  & \mbox{ for all } x, y, v, v' \mbox{ (Triangle Inequality)}\\
  \Lpm{x}{v} \leq \Lpm{y}{v}
  & \mbox{ for all } x, y, v \mbox{ such that } \Pref[v]{x}{y}
  \mbox{ (consistency)}\\
  \sum_v \Lpm{\LPC}{v} = 1 & \mbox{ (normalization)}\\
  \sum_v \Lpm{x}{v} \geq 1 & \mbox{ for all $x$ (optimality of \LPC)}\\
  \Lpm{x}{v} \geq 0 & \mbox{ for all } x, v.
\end{LP}

As is standard in the use of LPs for optimizing a ratio,
the normalization side-steps the issue of having to write a ratio:
for any worst-case metric, one could simply rescale all distances by a
constant so that the normalization holds --- this does not change
any ratios, and thus also not the distortion.

\subsection{An Efficient Optimal Mechanism} \label{sec:optimal}

As already observed in
\cite{anshelevich:bhardwaj:elkind:postl:skowron,goel:krishnaswamy:munagala},
the LP \eqref{eqn:primal-lp} can be leveraged to immediately yield an
instance-optimal polynomial-time mechanism for minimizing distortion,
as follows.
Given the voter preferences \PREF[v],
let \OptDist{\WINNER}{\LPC} denote the maximum LP objective of the
LP~\eqref{eqn:primal-lp} for the winner \WINNER and putative optimum
\LPC.
The distortion for \WINNER as a winner is then
$\WOptDist{\WINNER} = \max_{\LPC} \OptDist{\WINNER}{\LPC}$.
The mechanism returns as winner any candidate in
$\argmin_{\WINNER} \WOptDist{\WINNER}$.

Because the algorithm only involves solving $n^2$ linear programs,
it runs in polynomial time.
By definition (and correctness of the LP~\eqref{eqn:primal-lp}),
the distortion for a given vote profile \ALLPREFS
and winner \WINNER is \WOptDist{\WINNER};
thus, the mechanism does indeed minimize distortion.
Unfortunately, as also observed in
\cite{anshelevich:bhardwaj:elkind:postl:skowron},
it is not immediate from the mechanism and the LP formulation how to
bound the distortion for all vote profiles;
though \cite{goel:krishnaswamy:munagala} conjecture that the LP-based
algorithm guarantees distortion at most 3.

The dual program provides a very useful tool towards making the
LP-based algorithm combinatorial,
and for reducing an analysis of its distortion to simpler
combinatorial conjectures.
More generally (and perhaps importantly),
the dual program provides a general approach for bounding the metric
distortion of other voting rules, too.

\subsection{The Dual Linear Program} \label{sec:dual}
Rearranging the primal LP into normal form, taking the dual,
and switching the signs of the \DNo{x} variables (for clarity)
yields the following dual LP~\eqref{eqn:dual-lp}.
In this LP, the \DTe{v}{v'}{x}{y} are the dual variables for the
triangle inequality constraints,
\DCo{v}{x}{y} are the dual variables for the consistency constraints,
and the \DNo{x} are the dual variables for the
normalization/optimality constraints.

\begin{LP}[eqn:dual-lp]{Minimize}{\sum_x \DNo{x}}
\multicolumn{2}{l}{  \DNo{x}
  + \sum_{y: \Pref[v]{x}{y}} \DCo{v}{x}{y}
  - \sum_{y: \Pref[v]{y}{x}} \DCo{v}{y}{x}
  + \sum_{y, v'}
    \left(  \DTe{v}{v'}{x}{y} - \DTe{v}{v'}{y}{x}
          - \DTe{v'}{v}{x}{y} - \DTe{v'}{v}{y}{x} \right)}
  \\ \qquad \geq \;
  \begin{cases} 1 \mbox{ if } x = \WINNER \\
                0 \mbox{ if } x \neq \WINNER \end{cases} 
        & \mbox{ for all } v, x \\
  \DTe{v}{v'}{x}{y} \geq 0 & \mbox{ for all } v, v', x, y \\
  \DCo{v}{x}{y} \geq 0 & \mbox{ for all } v, x, y \\
  \DNo{x} \leq 0 & \mbox{ for all } x \neq \LPC.
\end{LP}

Notice that \DNo{\LPC} is in fact unconstrained,
due to the equality constraint in the normalization.

The advantage of studying the dual linear program instead of the
primal (or reasoning about the distortion directly) is that
it omits any reference to any metric space.
Rather than having to reason about all candidate metric spaces
consistent with given voting patterns, by weak duality,
we only have to exhibit one setting of the dual variables that yields
a small dual objective value.
Thus, our goal in analyzing a mechanism will be to show that for any
voter preferences \ALLPREFS, 
with a suitably chosen winner \WINNER, there is a setting of dual
variables giving a small objective value.

\subsection{Using the Dual by Exhibiting Flows}

The LP~\eqref{eqn:dual-lp} looks rather unwieldy,
mostly due to the terms involving the \DTe{v}{v'}{x}{y} variables.
However, by making some specific choices for these variables,
it can be interpreted as a flow\footnote{%
Some sources use the word ``flow'' only when there is a single source
and a single sink; here, we will have multiple sources and sinks.
We will still use the word ``flow'' in a generic sense.}
problem on a suitably defined graph,
with a somewhat unusual objective function.
This is captured by the following lemma:

\begin{lemma} \label{lem:dual-flow}
  Let $H = (U, E)$ be a directed graph with vertex set
  $U = \ALLVOTERS \times \ALLCANDS$, and edges defined as follows:
  \begin{itemize}
  \item Whenever \Pref[v]{x}{y},
    $E$ contains the directed edge $(v,x) \to (v,y)$.
    We call such edges \emph{preference edges}.
  \item For all $x$ and $v \neq v'$,
    $E$ contains the directed edge $(v,x) \to (v',x)$.
    We call such edges \emph{sideways edges}.
  \end{itemize}

  Let $f$ be a flow on $H$ such that exactly one unit of flow
  originates at the node $(v,\WINNER)$ for each voter $v$,
  and flow is only absorbed at nodes $(v,\LPC)$ for voters $v$.
  Define the cost of $f$ at voter $v$ to be
  $\FlCost{f}{v} = \sum_{e \text{ into } (v,\LPC)} f_e
  + \sum_{x \neq \LPC} \sum_{v' \neq v}
  (f_{(v',x) \to (v,x)} + f_{(v,x) \to (v',x)})$.

  Then, $\Cost{\WINNER} \leq \Cost{\LPC} \cdot \max_v \FlCost{f}{v}$.
\end{lemma}

The graph $H$ has two types of edges.
For any fixed voter $v$, the preference edges
$(v,x) \to (v,y)$ (over all candidate pairs $x, y$) exactly correspond
to $v$'s preference order.
For any fixed candidate $x$, the sideways edges
$(v,x) \to (v',x)$ (over all voter pairs $v, v'$) form a complete
directed graph.

The flow's cost function has two terms for each voter $v$.
The first is fairly standard in the study of multi-commodity flows:
the capacity required at the sink node $(v,\LPC)$ to be able
to absorb all of the flow.
The second one is rather non-standard: for each voter $v$,
there is an additional penalty term for all incoming and outgoing
flows of nodes $(v,x)$ for $x \neq \LPC$ along sideways edges.
In other words, using preference edges is much
less costly than using sideways edges:
the former just route flow, while the latter route the flow,
but also incur a cost penalty at both endpoints.

\begin{proof}
  Let $f$ be a flow with one unit of flow originating at each
  node $(v,\WINNER)$,
  such that flow is only absorbed at nodes $(v,\LPC)$.
  We define dual variables, and show that these dual variables are
  feasible.
  Then, we will obtain the statement of the lemma by weak LP duality. 

  For each triple $v,x,x'$, we set
  $\DCo{v}{x}{x'} = f_{(v,x) \to (v,x')}$.
  For each triple $v,v',x$, we set
  $\DTe{v}{v'}{x}{\LPC} = f_{(v,x) \to (v',x)}$;
  notice that we carefully choose \LPC as the additional candidate for
  the dual variable.
  Finally, we set $\DNo{\LPC} = \max_v \FlCost{f}{v}$.
  All other dual variables are set to 0.

  We now verify that this assignment satisfies all dual constraints.
  First, because $\DTe{v}{v'}{x}{y} = 0$ and $\DNo{y} = 0$
  whenever $y \neq \LPC$, 
  the dual constraints for all $x \neq \LPC$ are exactly circulation
  constraints; that is, they require that (at least) one unit of flow
  originate with $(v,\WINNER)$, and that flow be conserved
  (or appear) at each node $x \neq \LPC, x \neq \WINNER$.
  Thus, all of these constraints are satisfied for the given dual
  variable assignment.

  For pairs $(v,\LPC)$, the left-hand side of the dual constraint
  totals the flow into $(v,\LPC)$ along any edges
  (these are the \DCo{v}{x}{\LPC} variables and the
  \DTe{v'}{v}{\LPC}{\LPC} variables),
  as well as all the \DTe{v}{v'}{x}{\LPC} and \DTe{v'}{v}{x}{\LPC}
  variables, for all $v' \neq v$.
  By definition of the dual variables, this is exactly the flow
  into $(v,\LPC)$,
  plus the flow into and out of all nodes $(v,x)$ for $x \neq \LPC$
  along edges of the form $(v,x) \to (v',x)$ and $(v',x) \to (v,x)$.
  Thus, it is exactly \FlCost{f}{v}.
  Because we set $\DNo{\LPC} = \max_v \FlCost{f}{v}$,
  the dual constraints for all pairs $(v,\LPC)$ are also satisfied.

  Since we have a dual feasible solution of objective value
  $\DNo{\LPC} = \max_v \FlCost{f}{v}$,
  by weak duality, for every metric, the cost of the primal
  is at most $\max_v \FlCost{f}{v}$.
  This completes the proof.
\end{proof}

\section{Generalization of Distortion Bounds for Undominated Nodes}
\label{sec:uncovered}
As a corollary of Lemma~\ref{lem:dual-flow},
we obtain a strong generalization of
Theorem~7 in \cite{anshelevich:bhardwaj:postl}
and Lemma~3.7 of \cite{munagala:wang:improved}
(which are given below for comparison).
The most general version can be stated as follows:

\begin{corollary} \label{cor:paths}
  Let $x_1, x_2, \ldots, x_\ell$ be (distinct) candidates such that
  for each $i = 2, \ldots, \ell$,
  at least a $\PrefFrac{i} > 0$ fraction of
  voters prefer candidate $x_{i-1}$ over candidate $x_{i}$.
  Define $\lambda_1 = 1, \lambda_2 = \frac{2}{\PrefFrac{2}}-1$,
  and $\lambda_i = \frac{2}{\PrefFrac{i}}$ for $2 < i \leq \ell$.
  Let $\Lambda = \max_{S \subseteq \SET{1, \ldots, \ell}, S \text{ indep.}} 
  \sum_{i \in S} \lambda_i$.
  (Here, independence of a set $S$ of natural numbers means that the
  set contains no two consecutive numbers.)
  Then, $\Cost{x_1} \leq \Lambda \cdot \Cost{x_\ell}$.
\end{corollary}

\begin{proof}
  We define a flow $f$ and analyze its cost.
  For each $i$, we call the nodes $(v,x_i)$ (for all voters $v$)
  \emph{layer $i$}.
  Let $V_{i}$ be the set of voters $v$ with \Pref[v]{x_{i-1}}{x_i},
  with $V_1 := \ALLVOTERS$ for notational simplicity.

  We construct the flow layer by layer;
  our construction will ensure that each node $(v,x_i)$ with 
  $v \in V_i$ has exactly $\frac{m}{\SetCard{V_i}}$ units of flow entering.
  This holds in the base case $i=1$,
  because each node in layer 1 is the source node of one unit of flow.
  
  For the \Kth{i} step of the construction,
  we first route all the flow within layer $i$ using sideways edges,
  from nodes $(v,x_i)$ with $v \in V_i$
  to nodes $(v',x_i)$ with $v' \in V_{i+1}$.
  We then route it to nodes $(v', x_{i+1})$ in layer $i+1$ using
  preference edges.
  More specifically, to route the flow within layer $i$,
  we first consider voters $v \in V_i \cap V_{i+1}$.
  For those voters, 
  $\min(\frac{m}{\SetCard{V_i}}, \frac{m}{\SetCard{V_{i+1}}})$ units
  of flow simply stay at $(v,x_i)$.
  The node $(v,x_i)$ for such $v$ will have additional incoming flow
  from other nodes
  (if $\frac{m}{\SetCard{V_{i+1}}} > \frac{m}{\SetCard{V_{i}}}$)
  or additional outgoing flow to other nodes
  (if $\frac{m}{\SetCard{V_{i+1}}} < \frac{m}{\SetCard{V_{i}}}$).
  The remaining flow is routed arbitrarily using sideways edges
  from nodes $(v,x_i)$ with $v \in V_i$
  to nodes $(v',x_i)$ with $v' \in V_{i+1}$,
  of course ensuring that each such node $(v',x_i)$ has in total
  $\frac{m}{\SetCard{V_{i+1}}}$ units of flow entering.
  
  After this redistribution within layer $i$,
  each $(v,x_i)$ routes its flow to $(v,x_{i+1})$.
  Notice that this is always possible,
  because \Pref[v]{x_{i}}{x_{i+1}} for all $v \in V_{i+1}$.
  The construction is illustrated with an example in
  Figure~\ref{fig:flow-construction-uncovered}.
  
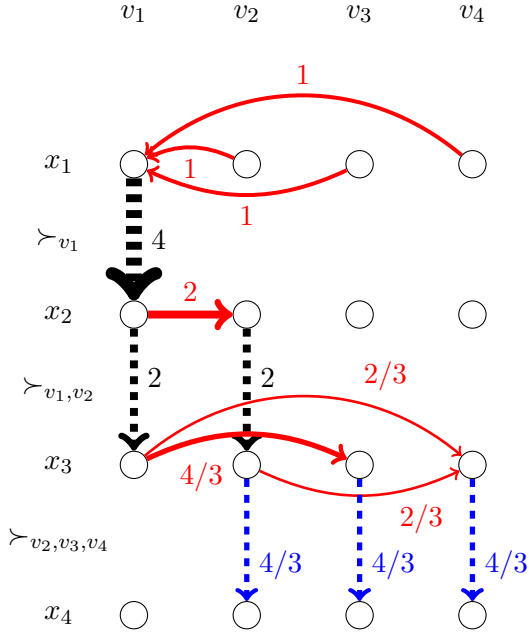
\begin{figure}[htb]
\begin{center}
  \begin{tikzpicture}[auto,active/.style={circle,draw=black}]
  \pgfsetxvec{\pgfpoint{0cm}{-1cm}}
  \pgfsetyvec{\pgfpoint{-1cm}{0cm}}

  \draw (0,4.5) node {$v_1$};
  \draw (0,3) node {$v_2$};
  \draw (0,1.5) node {$v_3$};
  \draw (0,0) node {$v_4$};

  \draw (2,5.5) node {$x_1$};
  \draw (3,5.5) node {$\PREF[v_1]$};
  \draw (4,5.5) node {$x_2$};
  \draw (5,5.5) node {$\PREF[v_1,v_2]$};
  \draw (6,5.5) node {$x_3$};
  \draw (7,5.5) node {$\PREF[v_2,v_3,v_4]$};
  \draw (8,5.5) node {$x_4$};

  \node[active] (v11) at (2,4.5) {};
  \node[active] (v12) at (2,3)   {};
  \node[active] (v13) at (2,1.5) {};
  \node[active] (v14) at (2,0) {};

  \node[active] (v21) at (4,4.5) {};
  \node[active] (v22) at (4,3) {};
  \node[active] (v23) at (4,1.5) {};
  \node[active] (v24) at (4,0) {};

  \node[active] (v31) at (6,4.5) {};
  \node[active] (v32) at (6,3) {};
  \node[active] (v33) at (6,1.5) {};
  \node[active] (v34) at (6,0) {};

  \node[active] (v41) at (8,4.5) {};
  \node[active] (v42) at (8,3) {};
  \node[active] (v43) at (8,1.5) {};
  \node[active] (v44) at (8,0) {};

  \draw [line width = 1.5pt,->,red] (v12) to [bend right=25] node {1} (v11);
  \draw [line width = 1.5pt,->,red] (v13) to [bend left=25] node {1} (v11);
  \draw [line width = 1.5pt,->,red] (v14) to [bend right=40] node[swap] {1} (v11);

  \draw [line width = 6pt,dashed,->] (v11) to node {4} (v21);

  \draw [line width = 3pt,->,red] (v21) to node {2} (v22);

  \draw [line width = 3pt,dashed,->] (v21) to node[pos=0.4] {2} (v31);
  \draw [line width = 3pt,dashed,->] (v22) to node[pos=0.4] {2} (v32);

  \draw [line width = 2pt,->,red] (v31) to [bend left = 25] node[swap,pos=0.1] {4/3} (v33);
  \draw [line width = 1pt,->,red] (v31) to [bend left = 40] node[pos=0.65] {2/3} (v34);
  \draw [line width = 1pt,->,red] (v32) to [bend right = 25] node[swap,pos=0.65] {2/3} (v34);
  
  \draw [line width = 2pt,dashed,->,blue] (v32) to node[pos=0.7] {4/3} (v42);
  \draw [line width = 2pt,dashed,->,blue] (v33) to node[pos=0.7] {4/3} (v43);
  \draw [line width = 2pt,dashed,->,blue] (v34) to node[pos=0.7] {4/3} (v44);
\end{tikzpicture}
\end{center}
\caption{An illustration of the flow construction.
  In the example, there are 4 voters and 4 relevant candidates,
  with voter preferences shown on the left.
  The preference fractions are
  $\PrefFrac{1} = 1/4, \PrefFrac{2} = 1/2, \PrefFrac{3} = 3/4$.
Sideways flows are shown in solid red,
while flow along preference edges is shown in dashed lines.
The dashed lines into nodes for candidate $x_4$ are shown in blue
(instead of black), to emphasize that they contribute to the objective
function.
The amount of flow is given numerically, and also shown using the
width of the lines/arcs.
Edges that are not used by the flow are not shown.
\label{fig:flow-construction-uncovered}}
\end{figure}

  We now analyze the cost associated with any fixed voter $v$.
  The cost has two components: the incoming flow at $(v,x_\ell)$
  (shown in blue in Figure~\ref{fig:flow-construction-uncovered}),
  and the cost associated with incoming/outgoing flow
  using sideways edges incident on $(v,x_i)$ for $i < \ell$
  (shown in red in Figure~\ref{fig:flow-construction-uncovered}).
  We begin with the incoming flow at $(v,x_\ell)$:
  if $v \in V_{\ell-1}$, the incoming flow is
  $\frac{m}{\SetCard{V_{\ell}}} = \frac{1}{\PrefFrac{\ell}}$;
  otherwise, it is 0.
  
  Next, we consider the cost associated with sideways edges.
  As a general guideline (subtleties will be discussed momentarily),
  when $v \in V_i$, the node $(v,x_i)$ has
  $\frac{m}{\SetCard{V_i}} = \frac{1}{\PrefFrac{i}}$ units of flow coming
  in along sideways edges,
  and the node $(v,x_{i+1})$ has the same amount of flow leaving
  along sideways edges.
  The associated cost of both together is $\frac{2}{\PrefFrac{i}}$.
  Two obvious exceptions are layers $i=1$ and $i=\ell-1$.
  For $i=1$, one unit of flow simply originates with $(v,x_1)$,
  resulting in no cost.
  For $i=\ell-1$, no sideways edge is used to route outgoing flow;
  however, this is compensated by the incoming flow at
  $(v,x_\ell)$ (discussed in the preceding paragraph),
  which adds the same cost term.

  However, simply adding up the bounds from the preceding paragraph
  over all steps $i$ with $v \in V_i$ is too pessimistic, because our
  flow construction avoids routing more flow than necessary when
  $v \in V_i \cap V_{i+1}$.
  A tighter bound is captured by the following lemma, proved below:

\begin{lemma} \label{lem:unimodal}
  Let $I$ be the set of all indices $i$ with $v \in V_i$.
  $I$ can be partitioned into disjoint intervals of integers
  $\SET{L_1, L_1 + 1, \ldots, R_1}, \SET{L_2, L_2 + 1, \ldots, R_2},
  \ldots, \SET{L_K, L_K+1, \ldots, R_K}$ (for some $K \geq 1$) such
  that:
  \begin{enumerate}
  \item For each $k$, there exists an index
    $M_k \in \SET{L_k, \ldots, R_k}$
  such that 
  $\PrefFrac{L_k} \geq \PrefFrac{L_k+1} \geq \cdots \geq \PrefFrac{M_k} > 0$
  and
  $\PrefFrac{M_k} \leq \PrefFrac{M_k+1} \leq \cdots \leq \PrefFrac{R_k}$;
  that is, the \PrefFrac{i} are monotone non-increasing from $L_k$ to $M_k$,
  and monotone non-decreasing from $M_k$ to $R_k$.
  \item The total cost of flow (both sideways flow and flow into
  $(v,x_\ell)$ in case $R_k = \ell$)
  associated with nodes $(v,x_i)$ with $L_k \leq i \leq R_k$
  is at most $\lambda_{M_k}$.
  \end{enumerate}
\end{lemma}

  To apply Lemma~\ref{lem:unimodal}, the key observation is that the
  set $\SET{M_1, M_2, \ldots, M_K}$ is independent, i.e., contains no
  two consecutive integers.
  If it did --- say, $i=M_k$ and $i+1=M_{k'}$ --- then both $i, i+1 \in I$.
  If $\PrefFrac{i+1} \leq \PrefFrac{i}$, this would contradict the
  maximality of $i$ in the definition of $M_k$;
  on the other hand, if $\PrefFrac{i+1} > \PrefFrac{i}$,
  then $i+1 \leq R_k$ by the definition of $R_k$, so it is impossible
  that $i+1 = M_{k'}$.

  Now, summing up the costs for each of the disjoint intervals,
  we obtain that the total cost of the flow at nodes associated with
  $v$ (both sideways flow and flow into $(v,x_{\ell})$) is at most
  $\sum_{k=1}^K \lambda_{M_k}$;
  because the set of $M_k$ is independent, this sum is at most $\Lambda$.
  Using Lemma~\ref{lem:dual-flow}, this completes the proof.
\end{proof}

\begin{extraproof}{Lemma~\ref{lem:unimodal}}
  We inductively define $L_k, R_k, M_k$ satisfying the first property,
  then show that they also guarantee the second property.
  For the base case, we set (for convenience) $R_0 := 0$.
  For the inductive step, focus on any $k \geq 1$.
  Define $L_{k} = \min \Set{i}{i \in I, i > R_{k-1}}$.
  (The construction terminates when there is no such $i$.)
  Let $M_k = \max \Set{i}{\SET{L_k, \ldots, i} \subseteq I,
    \PrefFrac{L_k} \geq \PrefFrac{L_k+1} \geq \cdots \geq \PrefFrac{i}}$.
  In words, $M_k$ is the largest index $i$ such that all indices
  between $L_k$ and $i$ are in $I$, and the \PrefFrac{j} values are
  monotone non-increasing all the way to $i$.
  Notice that because $M_k \in I$, we also have $\PrefFrac{M_k} > 0$.
  Now, let
  $R_k = \max \Set{i}{\SET{L_k, \ldots, i} \subseteq I,
    \PrefFrac{M_k} \leq \PrefFrac{M_k+1} \leq \cdots \leq \PrefFrac{i}}$.
  In words, $R_k$ is the largest index $i$ such that all indices
  between $M_k$ and $i$ (and thus also between $L_k$ and $i$)
  are in $I$, and the \PrefFrac{j} values are monotone non-decreasing
  from $M_k$ to $i$.
  This definition explicitly ensures that each interval
  $\SET{L_k, \ldots, R_k}$ is entirely contained in $I$,
  and satisfies the given monotonicity conditions.
  We now verify the second property.

  We first consider the case $M_k \notin \SET{1, 2, \ell-1}$,
  where $\lambda_{M_k} = 2/\PrefFrac{M_k}$.
  The important observation for the proof, also visible in
  Figure~\ref{fig:flow-construction-uncovered}, is that when
  $v \in V_i \cap V_{i+1}$, this eliminates sideways flow to and from
  nodes associated with $v$.
  Specifically, none of the nodes $(v,x_i)$ for $L_k \leq i < M_k$
  have outgoing flow along sideways edges
  (since all their flow stays for the next step).
  The incoming flow at $(v,x_i)$ along sideways edges is exactly
  $1/\PrefFrac{i} - 1/\PrefFrac{i-1}$
  (with $1/\PrefFrac{L_k-1} := 0$ for convenience);
  the remaining flow at $(v,x_i)$ is what is kept from step $i-1$.
  Similarly, none of the nodes $(v,x_i)$ for $M_k < i \leq R_k$ have
  incoming flow along sideways edges,
  since the node $(v,x_{i-1})$ has enough flow to meet all
  of the needs of $(v,x_i)$;
  the outgoing flow at such nodes $(v,x_i)$ along sideways edges is
  exactly $1/\PrefFrac{i} - 1/\PrefFrac{i+1}$
  (with $1/\PrefFrac{R_k+1} := 0$).
  Summing up all the incoming flows for $i = L_k, \ldots, M_k$
  (a telescoping series),
  and the outgoing flows for $i = M_k, \ldots, R_k$
  (another telescoping series),
  the total flow on sideways edges for all $(v,x_i)$ with
  $i \in \SET{L_k, \ldots, R_k}$ is at most $2/\PrefFrac{M_k} = \lambda_{M_k}$.
  Because $M_k < \ell$, there is no other cost associated with these nodes.
  Next, we consider the remaining cases $M_k \in \SET{1, 2, \ell-1}$.
  \begin{enumerate}
  \item If $M_k=1$, then by construction, $2 \notin I$;
    otherwise, the fact that $\PrefFrac{1} = 1 \geq \PrefFrac{2}$ would rule out
    setting $M_k=1$.
    Thus, the entire interval is just $\SET{L_1, \ldots, R_1} = \SET{1}$.
    Because there is no incoming sideways flow into $(v,x_1)$,
    the only cost is for one unit of outgoing sideways flow, i.e.,
    the cost is $1 = \lambda_1$.
  \item If $M_k=2$, then we must have $k=1$ and $L_k = 1$,
    because $1 \in I$ always by definition.
    We can directly apply the general analysis,
    except that we can subtract one unit of cost,
    the reason being (as in the case $M_k=1$)
    that there is no one unit of sideways flow into $(v,x_1)$.
    Thus, the total cost associated with the interval
    $\SET{L_1, \ldots, R_1}$ is at most $2/\PrefFrac{2} - 1 = \lambda_2$.
  \item If $M_k = \ell-1$, then $k=K$. There is no sideways flow out
    of $(v,x_{\ell-1})$ (and there are no nodes $(v,x_i)$ for $i >
    \ell-1$ to consider). Thus, the total cost of the sideways flows
    associated with $\SET{L_K, \ldots, R_K}$ is at most
    $1/\PrefFrac{M_k}$.
    On the other hand, in this case, there is also a cost of
    $1/\PrefFrac{M_k}$ for flow into $(v,x_{\ell})$
    (the blue flow in Figure~\ref{fig:flow-construction-uncovered});
    however, the total cost is still bounded by
    $2/\PrefFrac{M_k} = \lambda_{M_k}$.
  \end{enumerate}
  This shows that the bound holds for all cases of the interval.
\end{extraproof}

\subsection{Special Cases}
Lemma~3.7 of \cite{munagala:wang:improved} is the special case
of Corollary~\ref{cor:paths} with $\ell=3, x_1 = \WINNER, x_3 = \OPT$,
and $\PrefFrac{1} = \frac{3-\sqrt{5}}{2}, \PrefFrac{2} = \frac{\sqrt{5}-1}{2}$.
Our Corollary~\ref{cor:paths} then exactly recovers the bound of
$2+\sqrt{5}$.

When we have a uniform lower bound on the $\PrefFrac{i}$,
Corollary~\ref{cor:paths} can be simplified significantly.
(A direct proof of Corollary~\ref{cor:paths-uniform} would
also be simpler than the proof of the more general
Corollary~\ref{cor:paths}.)

\begin{corollary} \label{cor:paths-uniform}
  Let $x_1, x_2, \ldots, x_\ell$ be (distinct) candidates such that
  for each $i = 2, \ldots, \ell$, at least a $\PREFFRAC > 0$ fraction of
  the voters prefer candidate $x_{i-1}$ over candidate $x_{i}$.
  Then, if $\ell$ is even, 
  $\Cost{x_1} \leq (\frac{\ell}{\PREFFRAC} - 1) \cdot \Cost{x_\ell}$;
  if $\ell$ is odd,
  $\Cost{x_1} \leq (\frac{\ell-1}{\PREFFRAC} + 1) \cdot \Cost{x_\ell}$.
\end{corollary}

\begin{proof}
  We substitute $\PrefFrac{i} = \PREFFRAC$ for all $i$ in
  Corollary~\ref{cor:paths};
  then, we observe that for even $\ell$, the independent set of
  integers giving the largest sum is $\SET{2, 4, \ldots, \ell}$,
  while for odd $\ell$, it is $\SET{1, 3, \ldots, \ell}$.
\end{proof}

The asymmetry between even and odd $\ell$ disappears when
$\PREFFRAC = \half$ (i.e., in the case of the majority graph),
where the bound simply becomes $2\ell-1$.
The result thus strongly generalizes Theorem~7 in
\cite{anshelevich:bhardwaj:postl},
which is the special case of $\PREFFRAC = \half$ and $\ell=3$.

\section{Distortion of Ranked Pairs and Schulze} \label{sec:rp-schulze}
Corollary~\ref{cor:paths-uniform} allows us to pin down the
distortion of the Ranked Pairs and Schulze rules to 
within small constant factors.

\begin{corollary} \label{cor:rp-schulze}
Both the Ranked Pairs mechanism and the Schulze rule asymptotically
have distortion
at most $2 \sqrt{2} \cdot \sqrt{n} + o(\sqrt{n})$
and at least $\frac{\sqrt{2}}{2} \sqrt{n}$.
\end{corollary}

\begin{proof}
  We begin by proving the upper bounds.
  Let \WINNER be the candidate selected by the rule,
  and \OPT the optimum candidate.
  By Lemma~\ref{lem:rp-schulze-basic}, applied with $y = \OPT$,
  there exists a $p$ and a sequence of distinct candidates
  $x_1 = \WINNER, x_2, \ldots, x_{\ell} = \OPT$ with the property that
  for each $i$,
  at least a $p$ fraction of voters prefer $x_i$ over $x_{i+1}$,
  and at most a $p$ fraction of voters prefer \OPT over \WINNER.
  The existence of $x_1, \ldots, x_{\ell}$ with these properties is
  all that we need from the specific voting rules.
  The rest of the proof will be completely generic,
  and would thus also apply to any other voting rule satisfying
  Lemma~\ref{lem:rp-schulze-basic}.
  
  We consider two cases, depending on the value of $p$.
  The case $p \leq 1-\frac{1}{\sqrt{2n}}$ is easy.
  In this case, at least a $1-p \geq \frac{1}{\sqrt{2n}}$ fraction of
  voters prefer \WINNER over \OPT.
  Lemma~6 from \cite{anshelevich:bhardwaj:elkind:postl:skowron}
  states that if at least a $q$ fraction of voters prefer
  $x$ over $x'$, then
  $\Cost{x} \leq (1 + \frac{2(1-q)}{q}) \cdot \Cost{x'}$.
  Applying this lemma to \WINNER and \OPT,
  the distortion of \WINNER is at most
  $\frac{2}{1-p} - 1 \leq 2 \sqrt{2} \cdot \sqrt{n}$.

  When $p > 1-\frac{1}{\sqrt{2n}}$,
  we use Corollary \ref{cor:paths-uniform}.
  Let $\lambda = \Floor{\sqrt{n/2}}$, and
  $B = \Ceiling{\ell/\lambda}$.
  Because $\ell \leq n$,
  we get that $B \leq \sqrt{2n} + o(\sqrt{n})$.
  Consider the $B+1$ candidates
  $y_j := x_{j \lambda+1}$ for $j = 0, 1, \ldots, B-1$,
  and $y_B := x_\ell$.
  Because for each $i$,
  at least a $1 - \frac{1}{\sqrt{2n}}$ fraction of voters
  prefer $x_{i}$ to $x_{i+1}$,
  a union bound over the candidates
  $x_{j \lambda+1}, x_{j\lambda+2}, \ldots, x_{(j+1) \lambda}$
  shows that for each $j < B$, at least a
  $1 - \lambda \cdot \frac{1}{\sqrt{2n}} \geq \half$
  fraction of voters prefer $y_j$ over $y_{j+1}$.
  By using Corollary~\ref{cor:paths-uniform}
  (applied with $\PREFFRAC = \half$ and $\ell = B+1$),
  we obtain that
  \[
    \Cost{\WINNER} \; = \; \Cost{y_0}
    \; \leq \; (2B + 1) \cdot \Cost{y_B}
    \; \leq \; (2 \sqrt{2} \cdot \sqrt{n} + o(\sqrt{n})) \cdot \Cost{\OPT},
  \]
  so the distortion is upper-bounded by
  $2 \sqrt{2} \cdot \sqrt{n} + o(\sqrt{n})$.

  \medskip

  We now turn to a lower bound. Our lower-bound construction is a
  straightforward generalization of the construction that
  \cite{goel:krishnaswamy:munagala} used to show a lower bound of 5 on
  the distortion of the two rules.
  Let $m$ be given (assumed even), and set (with foresight\footnote{%
  The choice $B=m/2$ gives the tightest lower bound for this type of
  construction. Other choices work as well; for instance, $B=m$ gives
  a lower bound of $\frac{2}{3} \sqrt{n}$ instead of
  $\frac{\sqrt{2}}{2} \sqrt{n}$.}) $B=m/2$.
  Our construction has $m+2$ voters and $n=mB$ candidates
  $x_1, x_2, \ldots, x_{n}$.
  Voters $v=m+1, m+2$ have \Pref[v]{x_i}{x_{i+1}} for all $i$.
  We call this the \emph{default order}.
  To define the order (and later: distances) for voters
  $v=1, \ldots, m$, we define the following \emph{blocks} of
  consecutive (according to the default order) candidates.
  Block \Block{v}{b} for $1 \leq b < B$ consists of the $m$ candidates
  $\SET{x_{(b-1)m+v}, x_{(b-1)m+v+1}, \ldots, x_{bm+v-1}}$.
  Block \Block{v}{0} consists of the $v-1$ candidates
  $\SET{x_1, \ldots, x_{v-1}}$.
  (Notice that $\Block{1}{0} = \emptyset$.)
  Finally, block \Block{v}{B} consists of the $m+1-v$ candidates
  $\SET{x_{(B-1)m+v}, x_{(B-1)m+v+1}, \ldots, x_n}$.
  Voter $v=1, \ldots, m$ ranks the blocks in reverse order
  $\Block{v}{m}, \Block{v}{m-1}, \Block{v}{m-2}, \ldots, \Block{v}{0}$;
  within each block, $v$ ranks the candidates by the default order.
  An example of the block structure and ordering is shown in
  Figure~\ref{fig:RP-ranking}.
    
  \begin{figure}[htb]
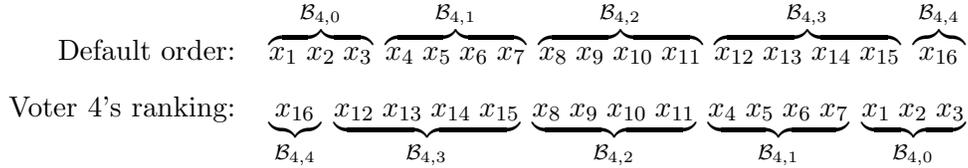

    \begin{align*}
      \text{Default order:} \quad &
      \overbrace{x_1 \; x_2 \; x_3}^{\Block{4}{0}} \;
      \overbrace{x_4 \; x_5 \; x_6 \; x_7}^{\Block{4}{1}} \;
      \overbrace{x_8 \; x_9 \; x_{10} \; x_{11}}^{\Block{4}{2}} \;
      \overbrace{x_{12} \; x_{13} \; x_{14} \; x_{15}}^{\Block{4}{3}} \;
      \overbrace{x_{16}}^{\Block{4}{4}}\\[1ex]
      \text{Voter 4's ranking:} \quad &
      \underbrace{x_{16}}_{\Block{4}{4}} \;
      \underbrace{x_{12} \; x_{13} \; x_{14} \; x_{15}}_{\Block{4}{3}} \;
      \underbrace{x_8 \; x_9 \; x_{10} \; x_{11}}_{\Block{4}{2}} \;
      \underbrace{x_4 \; x_5 \; x_6 \; x_7}_{\Block{4}{1}} \;
      \underbrace{x_1 \; x_2 \; x_3}_{\Block{4}{0}} \;
    \end{align*}
    
    \caption{An illustration of the blocks \Block{4}{b} and the
      preference ranking of voter 4. In this example, $m=B=4$.
      \label{fig:RP-ranking}}
  \end{figure}

  Intuitively, this means that on a ``global'' scale,
  voters $v=1, \ldots, m$ completely disagree with the default order,
  but locally, it looks like they agree.
  More formally, each voter $v$ disagrees with the default order only
  for pairs $(x_i, x_{i+1})$ when $i \equiv v-1 \mod m$.
  In particular, this means that for every pair $(x_i, x_{i+1})$, only
  one voter disagrees with the default order.
  For every pair $(x_i, x_j)$ with $i < j$, at least the voters
  $v = m+1, m+2$ have \Pref[v]{x_i}{x_j}, so the edges
  $(x_i, x_{i+1})$ have highest weight.
  In particular, this means that they are inserted first by Ranked
  Pairs, and hence Ranked Pairs chooses $\WINNER = x_1$.
  Similarly, $x_1$ has a path of width $1-\frac{1}{m+2}$ to each
  $x_i$, exceeding the width for any other node.
  Thus, the Schulze method also chooses $\WINNER = x_1$.

  We will now define a metric \DIST which is
  consistent with these rankings.
  Voters $v=m+1,m+2$ have distance $B$ to each candidate.\footnote{%
    To avoid tie breaking issues, one can easily perturb these and
    other equal distances by small $\epsilon$ values so that all
    distances are unique and the desired order is uniquely induced.
    We avoid doing so to not overload the proof with inessential notation.}
  Each voter $v = 1, \ldots, m$ has distance
  $\Dist{v}{x_i} = 2(B-b) + 1$ from all candidates $x_i \in B_{v,b}$.
  First, these distances explicitly ensure consistency with the
  voters' rankings.
  It remains to verify the triangle inequality.
  Consider two voters $v \neq v'$ and candidates $x_i \neq x_j$.
  We need to show that $\Dist{v}{x_i} \leq \Dist{v}{x_j} +
  \Dist{v'}{x_j} + \Dist{v'}{x_i}$.
  This is trivial if $v \in \SET{m+1,m+2}$,
  because $v$ has distance $B$ to all candidates.
  If $v' \in \SET{m+1,m+2}$, then the right-hand side contains two
  terms of $B$, and one term \Dist{v}{x_j} that is at least 1.
  Hence, the inequality holds.
  Finally, if both $v, v' \leq m$, then let $b,b'$ be the blocks such
  that $x_i \in \Block{v}{b}, x_i \in \Block{v'}{b'}$.
  The definition of the block structure ensures that
  $\Abs{b-b'} \leq 1$;
  in particular, $\Dist{v}{x_i} \leq \Dist{v'}{x_i} + 2$.
  Because $\Dist{v}{x_j} \geq 1, \Dist{v'}{x_j} \geq 1$,
  the triangle inequality again holds.

  Finally, we compute the social costs of $x_1$
  (the winner in Schulze and Ranked Pairs) and $x_n$.
  $x_1$ is in block 0 for all voters $v=2, \ldots, m$,
  and in block 1 for voter 1.
  Hence, his combined distance from these voters is
  $(m-1) (2B+1) + (2B-1) = 2Bm+m-2$.
  With the added distance of $B$ from each of $v=m+1,m+2$,
  the total cost of $x_1$ is $2Bm+2B+m-2 \geq 2Bm$.
  On the other hand, $x_n$ is at distance 1 from all voters
  $v=1, \ldots, m$ and at distance $B$ from $v=m+1,m+2$,
  for a total cost of $2B+m$.
  The ratio is thus at least
  $\frac{2Bm}{2B+m} \stackrel{B=m/2}{=} \frac{m}{2}
  \stackrel{n=m^2/2}{=} \frac{\sqrt{2}}{2} \cdot m$.
  This completes the proof.
\end{proof}

\begin{remark}
  The upper bound in Corollary~\ref{cor:rp-schulze}
  was a direct application of our flow-based framework.
  While the lower bound did not explicitly use the framework,
  the counter-example was in fact discovered after failed attempts to
  improve the upper bound. The failure to find ways to route flow very
  clearly suggested the types of rankings that were obstacles
  (i.e., reversed block structures).
  In turn, the distances were found essentially using the primal
  linear program.
\end{remark}

\section{A Candidate Algorithm for Distortion 3} \label{sec:combinatorial}
As a third application, we derive a purely combinatorial (i.e., not
LP-based) voting mechanism, which we conjecture to have distortion 3.
We show that this conjecture would follow from any of three
different-looking combinatorial conjectures we will formulate.

The point of departure for the derivation of the mechanism is
Corollary~\ref{cor:bipartite},
which simplifies Lemma~\ref{lem:dual-flow},
reducing it to a purely combinatorial property of a certain graph.
Corollary~\ref{cor:bipartite} was proved as Theorem~4.4
in \cite{munagala:wang:improved}, using a significantly more
complex proof.

For any two candidates $x, y$,
we consider the following bipartite graph \Bip{x}{y}
on the node set $(\ALLVOTERS, \ALLVOTERS)$;
that is, there is one node on the ``left'' for each voter $v$,
and one node on the ``right'' for each voter $v'$.
(We will use ``left'' and ``right'' to distinguish the two vertex sets.)
There is an edge $(v,v')$ if and only if there exists a candidate
$z \in \ALLCANDS$ ($z=x$ or $z=y$ are explicitly allowed)
such that \WPref[v]{x}{z} and \WPref[v']{z}{y}.

\begin{corollary} \label{cor:bipartite}
  Let $x \neq y$ be two candidates.
  If \Bip{x}{y} has a perfect matching,
  then $\Cost{x} \leq 3 \Cost{y}$.
\end{corollary}

\begin{proof}
  Assume that there is a perfect matching in \Bip{x}{y};
  for each voter $v$, let \Match{v} be the voter $v$ is matched with.
  By definition of \Bip{x}{y}, there is a candidate \Inter{v}
  such that \WPref[v]{x}{\Inter{v}} and \WPref[\Match{v}]{\Inter{v}}{y}.

  We now define the flow $f$ from each source node $(v, x)$.
  We route one unit of flow along the path
  $(v,x) \to (v,\Inter{v}) \to (\Match{v},\Inter{v}) \to (\Match{v},y)$.
  Notice that by definition of \Inter{v}, the first and third edge
  always exist.
  Also, if $\Inter{v} = x$, then the first two nodes are the same,
  and we omit the first edge.
  Similarly, if $\Match{v} = v$, we omit the second edge,
  and if $\Inter{v} = y$, we omit the third edge.

  This construction defines a valid flow, routing one
  unit of flow from each $(v,x)$ to some $(v',y)$ (for some $v'$).
  So it only remains to bound $\FlCost{f}{v} \leq 3$ for all $v$.

  Because \MATCH is a matching, there is exactly one unit of flow
  arriving at each node $(v,y)$.
  For a given voter $v$, let $v'$ be the unique voter with
  $\Match{v'} = v$.
  Then, the only two edges of the form $(v,z) \to (v',z)$ or
  $(v',z) \to (v,z)$ that can be used by $f$ are
  $(v,\Inter{v}) \to (\Match{v},\Inter{v})$ and
  $(v',\Inter{v'}) \to (v,\Inter{v'})$.
  Hence, the second part of the cost term \FlCost{f}{v} is at most 2,
  meaning that $\FlCost{f}{v} \leq 3$.
  The claim now follows by applying Lemma~\ref{lem:dual-flow}.
\end{proof}

Corollary~\ref{cor:bipartite} immediately suggests a natural
mechanism with distortion at most three,
which was also given as \textsc{MatchingUncovered} in
\cite{munagala:wang:improved}:

\begin{quote}
  \ABM:\\
  Find a candidate \WINNER such that for all other candidates $x$,
  the bipartite graph \Bip{\WINNER}{x} has a perfect matching.
\end{quote}

The mechanism \ABM sidesteps having to solve the $\Theta(n^2)$
LPs~\eqref{eqn:primal-lp}, 
instead solving $\Theta(n^2)$ bipartite matching problems.
The question then is whether such a candidate \WINNER actually exists.
We present three different conjectures,
each of which would imply the existence of \WINNER,
and thus, by Corollary~\ref{cor:bipartite},
that \ABM has distortion 3.

\subsection{The Candidate Comparison Graph \COMPG}
A key analysis tool in this section is a directed graph \COMPG on the
set of all candidates \ALLCANDS, which we call the \emph{\CandCompGr}.
\COMPG contains the directed edge $(y,x)$ if and only if the graph
\Bip{x}{y} does \emph{not} have a perfect matching.
In a sense, $y$ is a witness that $x$ would be a dangerous choice as
winner, since Corollary~\ref{cor:bipartite} would not apply
to bound the cost of $x$ when $y$ is the optimal candidate.%
\footnote{Of course, Corollary~\ref{cor:bipartite} is only a
  \emph{sufficient} condition, not a necessary one.
  So even when Corollary~\ref{cor:bipartite} cannot be applied, it is
  conceivable that \WINNER would achieve a distortion of 3.
  However, it is not clear which tool we could use to bound the
  distortion, which is why we focus only on the implications of
  Corollary~\ref{cor:bipartite} here.}
Conversely, any candidate \WINNER without incoming edges in \COMPG is a
safe choice as a winner, because Corollary~\ref{cor:bipartite} implies
a bound of 3 on its cost ratio to \OPT.

The following straightforward lemma captures that if we remove some
candidates, and leave each voter's ranking of the remaining candidates
unchanged, then the edges of the resulting graph \COMPGP
are a superset of the edges of the subgraph of \COMPG
induced by the remaining candidates.
We write $\ALLPREFS[\ALLCANDSP]$ for the vector of rankings
\PREFP[v], where each \PREFP[v] is the ranking \PREF[v],
restricted to candidates in \ALLCANDSP.
In other words, \PrefP[v]{x}{y} if and only if \Pref[v]{x}{y},
for all $x, y \in \ALLCANDSP$.

\begin{lemma} \label{lem:induced-graph}
  Let $(\ALLCANDS, \ALLPREFS)$ be a social choice instance,
  and $\ALLCANDSP \subseteq \ALLCANDS$.
  Let $\COMPG = \CompG{\ALLCANDS}{\ALLPREFS}$ and
  $\COMPGP = \CompG{\ALLCANDSP}{\ALLPREFS[\ALLCANDSP]}$.
  Then, $\COMPGP \supseteq \COMPG[\ALLCANDSP]$; 
  that is, \COMPGP contains all edges of the induced subgraph
  $\COMPG[\ALLCANDSP]$.
\end{lemma}

\begin{proof}
Consider two candidates $x,y \in \ALLCANDSP$ and their corresponding
bipartite graph \BipP{x}{y} on voters under the restricted instance.
By definition, \BipP{x}{y} contains the edge $(v,v')$ iff there exists
a candidate $z \in \ALLCANDSP$ such that \WPref[v]{x}{z} and \WPref[v']{z}{y}.
Thus, if \BipP{x}{y} contains the edge $(v,v')$, then so does the
bipartite graph \Bip{x}{y} for the larger/original instance
$(\ALLCANDS,\ALLPREFS)$.
In other words, the edges of \BipP{x}{y} are a subset of the edges of
\Bip{x}{y}.
Therefore, whenever \BipP{x}{y} contains a perfect matching,
so does \Bip{x}{y}.
Because edges in \COMPG (and \COMPGP) correspond to pairs that do
\emph{not} have bipartite matchings, the graph \COMPGP is a supergraph
of the induced subgraph $\COMPG[\ALLCANDSP]$.
\end{proof}



\subsection{Acyclicity of \COMPG} \label{sec:reformulation}
One sufficient condition for the existence of a source node in \COMPG
(i.e., a node without incoming edges) is for \COMPG to be acyclic.
This gives rise to our first conjecture,
which was also given as Conjecture~4.8 in
\cite{munagala:wang:improved}
(though it is expressed slightly differently there):

\begin{conjecture} \label{conj:acyclic}
  For every instance $(\ALLCANDS,\ALLPREFS)$
  the graph $\COMPG = \CompG{\ALLCANDS}{\ALLPREFS}$ is non-Hamiltonian.%
  \footnote{Recall that a directed graph is \emph{Hamiltonian} if it
    contains a directed cycle of all nodes.}
\end{conjecture}

Since our goal here is merely to show the existence of a source node
in \COMPG, it appears like overkill to aim for the ``stronger''
conjecture of being non-Hamiltonian/acyclic.
Despite appearances, Conjecture~\ref{conj:acyclic} is not in fact
stronger than the existence of a source node,
as we show in the following proposition:

\begin{proposition} \label{prop:acyclic-enough}
  \ABM succeeds \emph{on all inputs} if and only if
  Conjecture~\ref{conj:acyclic} is true.
\end{proposition}

Notice that the proposition does not say that whenever a \emph{specific}
instance violates Conjecture~\ref{conj:acyclic},
the algorithm will fail on \emph{that} instance.
As will be evident in the proof, it only implies that the algorithm
fails on \emph{some} (potentially different) instance. 

\begin{emptyproof}
\begin{enumerate}
\item For the first direction,
  if Conjecture~\ref{conj:acyclic} is true,
  then \COMPG is non-Hamiltonian for all inputs.
  We claim that this implies that in fact,
  \COMPG is \emph{acyclic} for all inputs.
  Suppose that we had an instance $(\ALLCANDS, \ALLPREFS)$
  for which \COMPG contains a directed cycle, say,
  $C = (x_1, x_2, \ldots, x_k, x_1)$ for some $k$.
  Let $\ALLCANDSP = \SET{x_1, \ldots, x_k}$.
  Consider the instance $(\ALLCANDSP, \ALLPREFS[\ALLCANDSP])$.
  By Lemma~\ref{lem:induced-graph},
  the graph \CompG{\ALLCANDSP}{\ALLPREFS[\ALLCANDSP]}
  is a supergraph of the induced subgraph
  $\COMPG[\ALLCANDSP]$, and must therefore contain a cycle including
  all vertices \ALLCANDSP, i.e.,
  \CompG{\ALLCANDSP}{\ALLPREFS[\ALLCANDSP]} is Hamiltonian.

  We have thus shown that for all instances, \COMPG is acyclic, meaning that
  for all inputs, \COMPG has a source node, which is a safe output for
  \ABM.
\item For the converse direction, assume that
  Conjecture~\ref{conj:acyclic} is false, and consider an instance
  $(\ALLCANDS, \ALLPREFS)$ for which \CompG{\ALLCANDS}{\ALLPREFS}
  is Hamiltonian.
  Because each node has at least one incoming edge,
  \CompG{\ALLCANDS}{\ALLPREFS} cannot have a source node. \QED
\end{enumerate}
\end{emptyproof}

\subsection{Distributions of Permutations}

We next derive a much simpler-looking --- but actually equivalent ---
conjecture, which is phrased only in terms of distributions of
permutations.
The key lemma for deriving this equivalent conjecture is the
following:

\begin{lemma} \label{lem:hall-application}
  \COMPG contains the edge $(y,x)$ if and only if
  there exists a set \EWit{x}{y} of candidates with
  $x \in \EWit{x}{y}$ and $y \notin \EWit{x}{y}$ such that

  \begin{align}
    \SetCard{\CandFirst{y}{\EWit{x}{y}}} +
    \SetCard{\CandLast{x}{\Compl{\EWit{x}{y}}}}
  & > m. \label{eqn:general-hall}
  \end{align}
\end{lemma}

\begin{emptyproof}
The proof relies on the well-known Hall Theorem:
\begin{theorem}[Hall's Bipartite Matching Theorem]
  Let $G = (X \cup Y, E)$ be a bipartite graph
  with $\SetCard{X} = \SetCard{Y}$.
  For any vertex set $S$, let \Neigh{S} denote the neighbors of $S$.
  $G$ has a perfect matching if and only if there is no contracting
  vertex set, i.e., no set $X' \subseteq X$ with
  $\SetCard{\Neigh{X'}} < \SetCard{X'}$.
\end{theorem}

  Fix a pair $x,y$ of candidates.
  By Hall's Theorem, $(y,x) \in \COMPG$
  (i.e., \Bip{x}{y} has no perfect matching)
  if and only if there is a contracting set of voters $V_{x,y}$.
  By definition of \COMPG, for every voter $v$ on the left,
  the neighborhood $\Neigh{v}$ consists of all voters $v'$ on the
  right such that there exists a candidate $z$
  with \WPref[v]{x}{z} and \WPref[v']{z}{y}. 

\begin{enumerate}
\item For the first direction,
  we assume that \COMPG contains the edge $(y,x)$.
  Let $V_{x,y}$ be a maximal contracting set.
  Let \EWit{x}{y} be the set of all candidates $z$ such that
  at least one voter $v \in V_{x,y}$ has \WPref[v]{x}{z}.
  Then, $\Neigh{V_{x,y}} = \Set{v'}{\mbox{there exists a } z \in
  \EWit{x}{y} \mbox{ with } \WPref[v']{z}{y}}$.
  This implies two things:
  \begin{itemize}
  \item $\Compl{\Neigh{V_{x,y}}} = \CandFirst{y}{\EWit{x}{y}}$
    is the set of all voters who rank $y$ strictly ahead of all of \EWit{x}{y};
    this follows directly from the preceding characterization.
  \item $V_{x,y} = \CandLast{x}{\Compl{\EWit{x}{y}}}$.
    The reason is that every voter $v \in V_{x,y}$ 
    ranks only candidates in \EWit{x}{y} (weakly) after $x$,
    so $V_{x,y} \subseteq \CandLast{x}{\Compl{\EWit{x}{y}}}$.
    Since $\Neigh{V_{x,y}} = \Neigh{\CandLast{x}{\Compl{\EWit{x}{y}}}}$,
    the set \CandLast{x}{\Compl{\EWit{x}{y}}} is also a candidate for
    a contracting set, and must equal $V_{x,y}$ by maximality of $V_{x,y}$.
  \end{itemize}
  By definition of \EWit{x}{y}, we always have $x \in \EWit{x}{y}$.
  Also, we always have $y \notin \EWit{x}{y}$,
  because any voter $v$ (on the left) with \Pref[v]{x}{y}
  has edges to all voters $v'$ on the right,
  and can therefore not be in a contracting set $V_{x,y}$.

  Next, we consider the cardinalities of the sets involved. Because
  \[
    m - \SetCard{\CandFirst{y}{\EWit{x}{y}}}
    \; = \; m - \SetCard{\Compl{\Neigh{V_{x,y}}}}
    \; = \; \SetCard{\Neigh{V_{x,y}}}
    \; \stackrel{\text{Hall}}{<} \; \SetCard{V_{x,y}}
    \; = \; \SetCard{\CandLast{x}{\Compl{\EWit{x}{y}}}},
  \]
  we have shown that there exists a set \EWit{x}{y} with
  $x \in \EWit{x}{y}$ and $y \notin \EWit{x}{y}$ such that
  $\SetCard{\CandFirst{y}{\EWit{x}{y}}}
    + \SetCard{\CandLast{x}{\Compl{\EWit{x}{y}}}}
    > m$.

  \item For the converse direction, assume that 
  there exists a set \EWit{x}{y} of candidates with
  $x \in \EWit{x}{y}$ and $y \notin \EWit{x}{y}$ such that
  $\SetCard{\CandFirst{y}{\EWit{x}{y}}}
  + \SetCard{\CandLast{x}{\Compl{\EWit{x}{y}}}}
  > m$.

  Let $V_{x,y} = \CandLast{x}{\Compl{\EWit{x}{y}}}$ be the set of all
  voters who rank $x$ behind all of $\Compl{\EWit{x}{y}}$.
  We will show that $V_{x,y}$ is contracting.
  Thereto, the important step is to characterize the neighborhood
  $\Neigh{V_{x,y}}$.
  By definition, it consists of all voters $v'$ such that
  there exists a voter $v \in V_{x,y}$ and a candidate $z$
  with \WPref[v]{x}{z} and \WPref[v']{z}{y}.
  Because each voter $v \in V_{x,y}$ ranks $x$ behind all of
  $\Compl{\EWit{x}{y}}$, the only potential candidates for $z$ are
  candidates in \EWit{x}{y}.
  In particular, no voter $v' \in \CandFirst{y}{\EWit{x}{y}}$
  can be in $\Neigh{V_{x,y}}$, implying that 
  $\CandFirst{y}{\EWit{x}{y}} \subseteq \Compl{\Neigh{V_{x,y}}}$.

  This implies that
  $\SetCard{\CandFirst{y}{\EWit{x}{y}}} \leq \SetCard{\Compl{\Neigh{V_{x,y}}}}$,
  so $\SetCard{\Neigh{V_{x,y}}} \leq m - \SetCard{\CandFirst{y}{\EWit{x}{y}}}$.
  And by definition of $V_{x,y}$, we also have
  $\SetCard{V_{x,y}} = \SetCard{\CandLast{x}{\Compl{\EWit{x}{y}}}}$.
  Together with the assumption that 
  $\SetCard{\CandFirst{y}{\EWit{x}{y}}}
  + \SetCard{\CandLast{x}{\Compl{\EWit{x}{y}}}}
  > m$, we get that
  $m < \SetCard{V_{x,y}} + (m - \SetCard{\Neigh{V_{x,y}}})$,
  implying that $V_{x,y}$ is contracting.
  By Hall's Theorem, \Bip{x}{y} has no perfect matching,
  so \COMPG contains the edge $(y,x)$. \QED
\end{enumerate}
\end{emptyproof}

\begin{remark} \label{rem:compg-subgraph}
As an easy corollary of Lemma~\ref{lem:hall-application}, notice that
the constraint \eqref{eqn:general-hall} implies that strictly more
than half of the voters prefer $y$ over $x$.
This is because $x \in \EWit{x}{y}$ and $y \notin \EWit{x}{y}$ implies that
all voters in \CandFirst{y}{\EWit{x}{y}} and all voters in 
\CandLast{x}{\Compl{\EWit{x}{y}}} rank $y$ ahead of $x$.
Because the combined cardinalities of the two sets add up to more than $m$,
by the Pigeon Hole Principle, at least one of the two sets
\CandFirst{y}{\EWit{x}{y}}, \CandLast{x}{\Compl{\EWit{x}{y}}} must contain
more than half of all voters.
This proves that \COMPG is a subgraph of the weak comparison graph.
\end{remark}

Based on Lemma~\ref{lem:hall-application}, we formulate the following
conjecture, and prove it equivalent to Conjecture~\ref{conj:acyclic}.

\begin{conjecture} \label{conj:permutation-distribution}
Let $\EWitS{1}, \EWitS{2}, \ldots, \EWitS{n} \subseteq \SET{1, \ldots, n}$
be arbitrary sets with $i \in \EWitS{i}$.
Define the following two indicator functions over elements $i$ and
total orders \PREF:
\begin{align}
  \IndPos{i}{\PREF} & = \begin{cases}
    1 & \text{if } \Pref{i+1}{\EWitS{i}} \\
    0 & \text{otherwise};
                \end{cases}
  & \IndNeg{i}{\PREF} & = \begin{cases}
    1 & \text{if } \Pref{\Compl{\EWitS{i}}}{i}\\
    0 & \text{otherwise};
                \end{cases}
\end{align}
here and for the rest of this proof, all additions/subtractions are
modulo $n$; that is, $n+1 := 1$ and $1-1 := n$.

Let \PREFDIST be any distribution over total orders \PREF of
$\SET{1, \ldots, n}$.
Then, there exists an $i$ such that
\begin{align}
  \Expect[\PREF \sim \PREFDIST]{\IndPos{i}{\PREF} + \IndNeg{i}{\PREF}} & \leq 1.
\end{align}
\end{conjecture}

\begin{proposition} \label{lem:first-two-conjectures}
Conjecture~\ref{conj:permutation-distribution} is true if and only if
Conjecture~\ref{conj:acyclic} is true.
\end{proposition}

\begin{proof}
We first show that Conjecture~\ref{conj:permutation-distribution}
implies Conjecture~\ref{conj:acyclic}, by proving the contrapositive.
Assume that \COMPG is Hamiltonian, containing a directed cycle
$x_n \to x_{n-1} \to \cdots \to x_1 \to x_n$ comprising all $n$ candidates.
By Lemma~\ref{lem:hall-application},
there exist sets $\EWitS{i} = \EWit{x_i}{x_{i+1}}$
(with $\EWitS{n} = \EWit{x_n}{x_1}$)
with $x_i \in \EWitS{i}, x_{i+1} \notin \EWitS{i}$,
such that for all $i=1, \ldots, n$, we have
  \begin{align}
    \SetCard{\CandFirst{x_{i+1}}{\EWitS{i}}}
    + \SetCard{\CandLast{x_i}{\Compl{\EWitS{i}}}}
  \; > \; m. \label{eqn:hall}
  \end{align}

We define a distribution \PREFDIST over rankings by drawing a
uniformly random voter from $\SET{1, \ldots, m}$,
and choosing this voter's ranking.
Because the distribution is uniform, we obtain that
\begin{align}
  \Expect[\PREF \sim \PREFDIST]{\IndPos{i}{\PREF}}
  & = \frac{1}{m} \cdot \SetCard{\CandFirst{x_{i+1}}{\EWitS{i}}},
  & \Expect[\PREF \sim \PREFDIST]{\IndNeg{i}{\PREF}}
  & = \frac{1}{m} \cdot \SetCard{\CandLast{x_i}{\Compl{\EWitS{i}}}}.
    \label{eqn:expectation-uniform}
\end{align}
The inequality~\eqref{eqn:hall} then implies that
$\Expect[\PREF \sim \PREFDIST]{\IndPos{i}{\PREF} + \IndNeg{i}{\PREF}} > 1$
for all $i$,
showing that Conjecture~\ref{conj:permutation-distribution} is
violated.

For the converse direction, assume that \PREFDIST is a distribution
over total orders of $\SET{1, \ldots, n}$ such that
\begin{align} 
  \Expect[\PREF \sim \PREFDIST]{\IndPos{i}{\PREF} + \IndNeg{i}{\PREF}}
  & > 1, \text{ for all $i$}. \label{eqn:permutation-violation}
\end{align}
Define $\delta := \min_i
\Expect[\PREF \sim \PREFDIST]{\IndPos{i}{\PREF} + \IndNeg{i}{\PREF}}$;
because the number of candidates is finite, $\delta$ is well-defined,
and $\delta > 1$.
Therefore, with sufficiently small changes to \PREFDIST,
we can ensure that the probability for each total order \PREF is a
rational number, while preserving all (strict)
inequalities~\eqref{eqn:permutation-violation}.
Once the probabilities are all rational, we can write them with a
common denominator, meaning that we can define \PREFDIST as a uniform
distribution over a finite multi-set of rankings;
in turn, we can consider these rankings as voters.

Because the distribution is uniform over voters,
we can apply the characterization~\eqref{eqn:expectation-uniform} to
conclude that 
$\SetCard{\CandFirst{x_{i+1}}{\EWitS{i}}}
+ \SetCard{\CandLast{x_i}{\Compl{\EWitS{i}}}}
> m$ for all candidates $i$.
By Lemma~\ref{lem:hall-application},
applied to each pair $(x_i, x_{i+1})$,
the graph \COMPG contains each edge $(x_{i+1}, x_i)$,
so \COMPG is Hamiltonian.
\end{proof}

\subsection{A Graph-Theoretic Reformulation}

Our attempts to prove Conjecture~\ref{conj:permutation-distribution}
(so far unsuccessful) have been based on proofs by contradiction.
The assumed constraints \eqref{eqn:permutation-violation} prescribe
several constraints on rankings that must hold simultaneously;
using transitivity, this leads to a contradiction by forcing
preferences to contain cycles.
The essence of this approach is captured by another conjecture.
To formulate it, we define the following class of directed graphs,
which we term \emph{Constraint-Choice Graphs}.

\begin{definition}[Constraint-Choice Graph]
  \label{def:constraint-choice-graph}
  Let $\CandNodes{n} = \SET{\CandNode{1}, \ldots, \CandNode{n}},
  \SetNodes{n} = \SET{\SetNode{1}, \ldots, \SetNode{n}},
  \ComplNodes{n} = \SET{\ComplNode{1}, \ldots, \ComplNode{n}}$ 
  be three disjoint sets of nodes.
  A \emph{constraint-choice graph} for a given $n$ contains $3n$ nodes
  $U_n = \CandNodes{n} \cup \SetNodes{n} \cup \ComplNodes{n}$, 
  and the following edges:
  \begin{itemize}
  \item For each $i$, it contains the directed edges%
\footnote{As before, we define $1-1:= n$ and $n+1 := 1$.}
  $(\CandNode{i}, \SetNode{i-1}), (\CandNode{i}, \ComplNode{i-1}), 
   (\SetNode{i}, \CandNode{i}), (\ComplNode{i}, \CandNode{i})$.
  \item For each $i, j$ with $j \neq i, j \neq i-1$, it contains
    exactly one of the two directed edges 
    $(\SetNode{j}, \CandNode{i}), (\CandNode{i}, \ComplNode{j})$.
  \end{itemize}
\end{definition}

An example of a constraint choice graph is shown in
Figure~\ref{fig:constraint-choice-graph}.
The edges listed first in Definition~\ref{def:constraint-choice-graph}
are shown in solid black, while the edges listed second are shown in
dashed red lines.

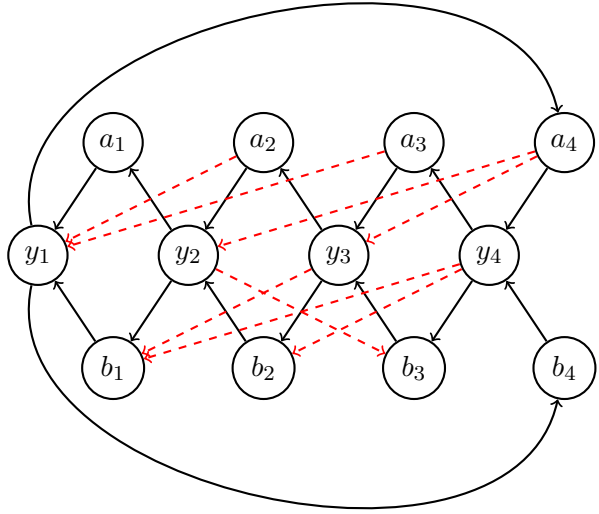
\begin{figure}[htb]
\begin{center}
\begin{tikzpicture}[auto,thick,->,active/.style={circle,draw=black}]

  \node[active] (y1) at (1,0) {\CandNode{1}};
  \node[active] (y2) at (3,0) {\CandNode{2}};
  \node[active] (y3) at (5,0) {\CandNode{3}};
  \node[active] (y4) at (7,0) {\CandNode{4}};

  \node[active] (a1) at (2,1.5) {\SetNode{1}};
  \node[active] (a2) at (4,1.5) {\SetNode{2}};
  \node[active] (a3) at (6,1.5) {\SetNode{3}};
  \node[active] (a4) at (8,1.5) {\SetNode{4}};

  \node[active] (b1) at (2,-1.5) {\ComplNode{1}};
  \node[active] (b2) at (4,-1.5) {\ComplNode{2}};
  \node[active] (b3) at (6,-1.5) {\ComplNode{3}};
  \node[active] (b4) at (8,-1.5) {\ComplNode{4}};

  \draw (a1) to (y1);
  \draw (b1) to (y1);
  \draw (y1) to [bend left = 90] (a4);
  \draw (y1) to [bend right = 90] (b4);
  \draw (a2) to (y2);
  \draw (b2) to (y2);
  \draw (y2) to (a1);
  \draw (y2) to (b1);
  \draw (a3) to (y3);
  \draw (b3) to (y3);
  \draw (y3) to (a2);
  \draw (y3) to (b2);
  \draw (a4) to (y4);
  \draw (b4) to (y4);
  \draw (y4) to (a3);
  \draw (y4) to (b3);

  \begin{scope}[red,dashed]
    \draw (y3) to (b1);
    \draw (y4) to (b1);
    \draw (y4) to (b2);
    \draw (a2) to (y1);
    \draw (a3) to (y1);
    \draw (y2) to (b3);
    \draw (a4) to (y3);
    \draw (a4) to (y2);
  \end{scope}
\end{tikzpicture}
\end{center}
\caption{An illustration of a constraint-choice graph for $n=4$ candidates.
  The graph depicted here corresponds to the sets
  $\EWitS{1} = \SET{1}, \EWitS{2} = \SET{1,2}, \EWitS{3} = \SET{1,3},
  \EWitS{4} = \SET{2,3,4}$ in the construction of the proof of
  Proposition~\ref{prop:choice-graph}.
\label{fig:constraint-choice-graph}}
\end{figure}








\begin{conjecture} \label{conj:constraint-choice-graph}
  For every $n$ and every constraint choice graph $G_n$ of $3n$ nodes,
  there exists a non-empty index set $S \subseteq \SET{1, \ldots, n}$
  with the following property:
  For every vertex set 
  $T \subseteq \Set{\SetNode{i}, \ComplNode{i}}{i \in S}$ 
  of size $\SetCard{T} > \SetCard{S}$, the induced subgraph
  $G_n[\CandNodes{n} \cup T]$ contains a directed cycle.
\end{conjecture}  

\begin{remark}
  Notice that the conjecture indeed talks about the subgraph induced
  by \emph{all} nodes \CandNode{i} (not just those with indices in $T$),
  in addition to at least $\SetCard{S} + 1$ nodes from among the
  \SetNode{i}, \ComplNode{i} with $i \in S$.
\end{remark}

\begin{proposition} \label{prop:choice-graph}
  If Conjecture~\ref{conj:constraint-choice-graph} is true,
  then Conjecture~\ref{conj:permutation-distribution} is true.
\end{proposition}

\begin{proof}
We prove the contrapositive, and assume that 
Conjecture~\ref{conj:permutation-distribution} is false; that is,
Inequality~\eqref{eqn:permutation-violation} holds for all $i$.

Given the (assumed) sets
$\EWitS{1}, \EWitS{2}, \ldots, \EWitS{n} \subseteq \SET{1, \ldots, n}$,
we define the following constraint-choice graph $G_n$.
It contains the $3n$ nodes
$\CandNodes{n} \cup \SetNodes{n} \cup \ComplNodes{n}$,
and all the edges that are prescribed by
Definition~\ref{def:constraint-choice-graph};
in addition, if $i \in \EWitS{j}$, then $G_n$ contains the edge
$(\SetNode{j}, \CandNode{i})$;
otherwise, it contains the edge $(\CandNode{i}, \ComplNode{j})$.
This completes the definition of $G_n$.
An example for specific sets \EWitS{j} is shown in
Figure~\ref{fig:constraint-choice-graph}.
To gain intuition for the following proof, the reader is encouraged to
think of \CandNode{i} as corresponding to candidate $x_i$,
of \SetNode{i} as corresponding to \EWitS{i},
and of \ComplNode{i} as corresponding to \Compl{\EWitS{i}}.

To prove that $G_n$ violates
Conjecture~\ref{conj:constraint-choice-graph},
consider an arbitrary non-empty set of indices
$S \subseteq \SET{1, \ldots, n}$.
By linearity of expectations,
\eqref{eqn:permutation-violation} implies that 
$\Expect[\PREF \sim \PREFDIST]{\sum_{i \in S} (\IndPos{i}{\PREF} + \IndNeg{i}{\PREF})}
> \SetCard{S}$. 
Because the maximum is at least the average,
this implies that there exists some ranking \PREF[S] with 
$\sum_{i \in S} (\IndPos{i}{\PREF[S]} + \IndNeg{i}{\PREF[S]}) > \SetCard{S}$,
and because the quantity on the left-hand side is integral,
we can strengthen this inequality to
\begin{align}
  \sum_{i \in S} (\IndPos{i}{\PREF[S]} + \IndNeg{i}{\PREF[S]})
  & \geq \SetCard{S} + 1.
    \label{eqn:permutation-violation-set}
\end{align}

Define the node set
$T_S := \Set{\SetNode{i}}{\IndPos{i}{\PREF[S]} = 1} \cup
     \Set{\ComplNode{i}}{\IndNeg{i}{\PREF[S]} = 1}$.
By Inequality~\eqref{eqn:permutation-violation-set},
$T_S$ contains at least $\SetCard{S} + 1$ nodes from
$\Set{\SetNode{i}, \ComplNode{i}}{i \in S}$.
We will show that the induced subgraph $G_n[T_S \cup \CandNodes{n}]$ is acyclic;
since we show this for every $S$, it proves that $G_n$ violates
Conjecture~\ref{conj:constraint-choice-graph}.

To show that $G_n[T_S \cup \CandNodes{n}]$ is acyclic,
we use a proof by contradiction,
and assume that $G_n[T_S \cup \CandNodes{n}]$ contains a cycle $C$.
Because edges only go between nodes \CandNode{i} and either
\SetNode{j} or \ComplNode{j}, this cycle must alternate nodes
\CandNode{i} with nodes \SetNode{j} or \ComplNode{j}.
Let $\SetCard{C} = 2k$, and let $i_1, i_2, \ldots, i_k$ be such that
the order of nodes \CandNode{i} in $C$ is
$\CandNode{i_1}, \CandNode{i_2}, \ldots, \CandNode{i_k}, \CandNode{i_1}$.
Fix some $\ell \in \SET{1, \ldots, k}$.
Between \CandNode{i_\ell} and \CandNode{i_{\ell+1}}
(here, $k+1 := 1$), the cycle must visit either some node
$\SetNode{j} \in T_S$ or some node $\ComplNode{j} \in T_S$.
We distinguish two cases:

\begin{itemize}
\item If the intermediate node is \SetNode{j},
  observe first that by Definition~\ref{def:constraint-choice-graph},
  the only incoming edge to \SetNode{j} is
  $(\CandNode{j+1}, \SetNode{j})$, so $i_\ell = j+1$.
  The specific constraint-choice graph $G_n$ defined in this proof
  includes outgoing edges from \SetNode{j} to exactly
  the \CandNode{i} with $i \in \EWitS{j}$;
  notice that this includes the case of the edge
  $(\SetNode{j}, \CandNode{j})$, because $j \in \EWitS{j}$.
  In particular, it applies to the edge
  $(\SetNode{j}, \CandNode{i_{\ell+1}})$,
  implying that $i_{\ell+1} \in \EWitS{j}$.

  Because $\SetNode{j} \in T_s$, we have that
  $\IndPos{i}{\PREF[S]} = 1$, meaning that under \PREF[S],
  candidate $i_{\ell} = j+1$ precedes all candidates in \EWitS{j};
  this includes, in particular, the candidate $i_{\ell+1}$.
  In summary, we have inferred that
  \Pref[S]{i_{\ell}}{i_{\ell+1}}.
\item If the intermediate node is \ComplNode{j},
  observe that by Definition~\ref{def:constraint-choice-graph},
  the only outgoing edge from \ComplNode{j} is
  $(\ComplNode{j}, \CandNode{j})$, implying that $i_{\ell+1} = j$.
  For the particular graph $G_n$ defined in this proof, the incoming edge
  $(\CandNode{i_{\ell}}, \ComplNode{j})$ exists exactly when
  $i_{\ell} \notin \EWitS{j}$.

  Because $\ComplNode{j} \in T_s$, we have that
  $\IndNeg{j}{\PREF[S]} = 1$, so under \PREF[S],
  candidate $j = i_{\ell+1}$ is ranked after all candidates in
  $\Compl{\EWitS{j}}$.
  By the argument of the preceding paragraph,
  the set \EWitS{j} includes, in particular, the candidate $i_{\ell}$.
  In summary, we have again inferred that
  \Pref[S]{i_{\ell}}{i_{\ell+1}}.
\end{itemize}

Thus, we have derived that \Pref[S]{i_{\ell}}{i_{\ell+1}}
for each $\ell = 1, \ldots, k$.
By transitivity, this results in a cycle in \PREF[S],
a contradiction to it being a ranking.
\end{proof}

\begin{remark}
Conjecture~\ref{conj:constraint-choice-graph} is sufficiently clean
and combinatorial that it can be verified by hand for $n \leq 4$.
An exhaustive computer search for $n \leq 7$ has verified
the conjecture for all such $n$, recovering and extending
computer-assisted results in Theorem~4.11 in
\cite{munagala:wang:improved} (although we did not include some ranges
for $n > 7$ which \cite{munagala:wang:improved} handle with a
restricted number of voters).
Unfortunately, because it basically involves a search over all
possibilities of $n$ subsets \EWitS{i} of $n$ elements
(represented as graph edge choices), the running time scales roughly
as $2^{n^2}$; from $n=6$ to $n=7$, the running time increases from
less than a minute to roughly a day.
Hence, even $n=8$ is likely out of reach.
But the computer search is encouraging in terms of trying to prove the
conjecture (rather than disproving it).
\end{remark}

\section{Conclusions} \label{sec:conclusions}
Our work raises a very obvious question: prove (or possibly disprove)
the conjectures stated in Section~\ref{sec:combinatorial}.
Based on exhaustive computer search, it seems more likely that the
conjectures are true, and the \ABM mechanism in fact is always able to
find a candidate with distortion at most 3.

Going beyond these conjectures, we believe that the duality-based
framework may be useful for bounding the performance of other voting
mechanisms, in particular, those that may miss information on parts of
voters' ranking.
For instance, such a situation can occur in the setting of
\cite{DistortionCommunication}, where voters can only name the
candidates in a subset of positions on their ballot,
rather than giving a complete ranking.
The analysis of a mechanism proposed in \cite{DistortionCommunication}
becomes much simpler (and tighter) using the techniques developed here.

While we have only studied deterministic mechanisms here,
the framework can also be extended to randomized mechanisms.
When the mechanism selects a candidate $x$ with probability
$q_x$, an upper bound can be obtained by bounding a flow that
inserts $q_x$ units of flow at each of the nodes $(v,x)$,
which again have to be routed to \LPC.

It is conceivable that duality-based approaches similar to the one we
developed could be helpful for the analysis of mechanisms for other
problems in the cardinal/ordinal framework:
a worst-case metric for a given input can often be characterized in
terms of a linear program,
and the dual may in general lead to a framework for proving upper
bounds on the performance of a chosen mechanism.

\subsubsection*{Acknowledgements}
The author would like to thank
Elliot Anshelevich, Sid Banerjee, Shaddin Dughmi, Bobby Kleinberg and Kai Wang
for useful conversations and pointers, and anonymous reviewers for
useful feedback.

\bibliographystyle{plain}
\bibliography{names,conferences,publications,bibliography}

\end{document}